\documentclass[aos]{imsart}
\RequirePackage{amsthm,amsfonts,amssymb}
\RequirePackage[numbers,sort]{natbib}
\RequirePackage[colorlinks,citecolor=blue,urlcolor=blue]{hyperref}
\RequirePackage{graphicx}

\usepackage[reqno]{amsmath}

\usepackage{enumerate}% allows enumerate with latin characters
\usepackage{dsfont}% allows doublestroke symbols with \mathds{}
\usepackage[usenames,dvipsnames]{color}
\usepackage{mathrsfs}
\usepackage{placeins}
\usepackage{graphicx}
\usepackage{tikz}
\usetikzlibrary{arrows}
\usepackage{url}
\usepackage{comment}
\usepackage{multicol}
\usepackage{cleveref}
\usepackage{stmaryrd}

\usepackage{natbib}

\renewcommand{\thetable}{\arabic{section}.\arabic{table}}
\numberwithin{equation}{section}

\theoremstyle{plain} %default (text in italic)
\newtheorem{theorem}{Theorem}[section]
\newtheorem{lemma}[theorem]{Lemma}
\newtheorem{proposition}[theorem]{Proposition}
\newtheorem{corollary}[theorem]{Corollary}

\theoremstyle{remark}
\newtheorem{definition}[theorem]{Definition}

\newtheorem{example}[theorem]{Example}

\newtheorem{remark}[theorem]{Remark}

\def\tablename{\footnotesize Table}

\newcommand{\bthe}{\begin{theorem}}
\newcommand{\ethe}{\end{theorem}}

\newcommand{\ben}{\begin{enumerate}}
\newcommand{\een}{\end{enumerate}}

\newcommand{\bit}{\begin{itemize}}
\newcommand{\eit}{\end{itemize}}

\newcommand{\beq}{\begin{equation}}
\newcommand{\eeq}{\end{equation}}

\newcommand{\ble}{\begin{lemma}}
\newcommand{\ele}{\end{lemma}}

\newcommand{\bde}{\begin{definition}\rm}
\newcommand{\ede}{\halmos\end{definition}}

\newcommand{\bco}{\begin{corollary}}
\newcommand{\eco}{\end{corollary}}

\newcommand{\bpr}{\begin{proposition}}
\newcommand{\epr}{\end{proposition}}

\newcommand{\brem}{\begin{remark}\rm}
\newcommand{\erem}{\end{remark}}

\newcommand{\bproof}{\begin{proof}}
\newcommand{\eproof}{\end{proof}}

\newcommand{\bexam}{\begin{example}\rm}
\newcommand{\eexam}{\end{example}}

\newcommand{\bfi}{\begin{fig}}
\newcommand{\efi}{\end{fig}}

\newcommand{\btab}{\begin{tab}}
\newcommand{\etab}{\end{tab}}

\newcommand{\beao}{\begin{eqnarray*}}
\newcommand{\eeao}{\end{eqnarray*}\noindent}

\newcommand{\balo}{\begin{align*}}
\newcommand{\ealo}{\end{align*}}

\newcommand{\balm}{\begin{align}}
\newcommand{\ealm}{\end{align}\noindent}

\newcommand{\beam}{\begin{eqnarray}}
\newcommand{\eeam}{\end{eqnarray}\noindent}

\newcommand{\barr}{\begin{array}}
\newcommand{\earr}{\end{array}}

\renewcommand\P{\mathbb{P}}
\newcommand{\R}{\mathbb{R}}

\def\RV{\mathcal{RV}}

\def\binfty{\boldsymbol \infty}

\def\bzero{\boldsymbol 0}
\def\bone{\boldsymbol 1}

\def\bX{\boldsymbol X}
\def\bY{\boldsymbol Y}
\def\bZ{\boldsymbol Z}

\def\bc{\boldsymbol c}

\def\bh{\boldsymbol h}
\def\bl{\boldsymbol l}

\def\bq{\boldsymbol q}
\def\bw{\boldsymbol w}
\def\bx{\boldsymbol x}
\def\by{\boldsymbol y}
\def\bz{\boldsymbol z}

\def\hg{\hat{\gamma}}

\def\balpha{\boldsymbol \alpha}
\def\balpha{\boldsymbol \alpha}

\def\btheta{\boldsymbol \theta}
\def\bkappa{\boldsymbol \kappa}

\newcommand{\stp}{\stackrel{P}{\rightarrow}}
\newcommand{\std}{\stackrel{d}{\rightarrow}}

\newcommand{\RVGC}{\text{RVGC}}
\newcommand{\PGC}{\text{P-GC}}

\DeclareMathOperator*{\argmax}{arg\,max}
\newcommand{\ov}{\overline}

\newcommand{\vague}{\stackrel{\lower0.2ex\hbox{$\scriptscriptstyle
                    \it{v} $}}{\rightarrow}}
\newcommand{\weak}{\stackrel{\lower0.2ex\hbox{$\scriptscriptstyle
                    \it{w} $}}{\rightarrow}}
\newcommand{\what}{\stackrel{\lower0.2ex\hbox{$\scriptscriptstyle
                    \it{\hat{w}} $}}{\rightarrow}}
\newcommand{\eqdis}{\stackrel{\lower0.2ex\hbox{$\scriptscriptstyle
                    \mathrm{d}$}}{=}}
\newcommand{\distr}{\stackrel{\lower0.2ex\hbox{$\scriptscriptstyle
                    \it{d} $}}{\rightarrow}}
\allowdisplaybreaks %% display breaks for align
\definecolor{darkgreen}{RGB}{0,139,0}

\begin{document}

\begin{frontmatter}

%\title{Modeling and estimation for heavy-tailed data\\ with Gaussian dependence}
\title{Inference for heavy-tailed data with Gaussian dependence}
\runtitle{Heavy-tailed data with Gaussian dependence}

\begin{aug}
  \author[A]{\fnms{Bikramjit} \snm{Das}\ead[label=e1]{bikram@sutd.edu.sg}\orcid{0000-0002-6172-8228}}
   % \and
  % \author[B]{\fnms{Vicky} \snm{Fasen-Hartmann}\ead[label=e2]{vicky.fasen@kit.edu}\orcid{0000-0002-5758-1999}}
%\thanksref{t1}\thankstext{t1}
 \address[A]{Engineering Systems and Design, Singapore University of Technology and Design  \printead[presep={,\ }]{e1}}

 % \address[B]{Institute of Stochastics, Karlsruhe Institute of Technology\printead[presep={,\ }]{e2}}

%  \thanksref{T1}
%  \thankstext{T1}{}

  \runauthor{B. Das}% and  V. Fasen-Hartmann}
\end{aug}

\begin{abstract}
We consider a model for multivariate data with heavy-tailed marginal distributions and a Gaussian dependence structure. The different marginals in the model are allowed to have non-identical tail behavior in contrast to most popular modeling paradigms for multivariate heavy-tail analysis. \linebreak Despite being a practical choice, results on parameter estimation and inference under such models remain limited.
In this article, consistent estimates for both  marginal tail indices and the Gaussian correlation parameters for such models are provided and asymptotic normality of these estimators are established. The efficacy of the estimation methods are exhibited using  \linebreak extensive simulations and then they are applied to real data sets from insurance claims, internet traffic, and, online networks.
\end{abstract}



%\begin{keyword}[class=JEL]
%\kwd[Primary ]{C14}
%\kwd{C61}
%\kwd{C63}
%\kwd{D81}
%\kwd{G11} 
%\kwd{G22}
%\iffalse
%\kwd[Primary ]{60F10}
%\kwd{60G50}
%\kwd{60G70}
%\kwd[; secondary ]{60B10}
%\kwd{62G32}
%\fi
%\end{keyword}

\end{frontmatter}

%{\rmfamily \tableofcontents}

\section{Introduction}\label{sec:intro} 

%Consider independent and identically distributed (i.i.d.)  random vectors $\bX_1, \ldots, \bX_n, \bX \sim F$ in $\R^d$ where the marginal distributions of $\bX_1$ are heavy-tailed in nature. 
%Such  examples are quite common for risk assessment in various applications, for example in evaluating tail risk measures like value-at-risk, expected shortfall and their multidimensional variants in finance and insurance (\citet{artzner:delbaen:eber:heath:1999, mcneil:frey:embrechts:2015, adrian:brunnermeier:2016}).
Empirical evidence indicates that the behavior of tail probability distribution of variables in many  applications are  roughly power-law-like, or subexponential in nature; this has been observed in  finance (\citet{mandelbrot:1963, ibragimov:prokhorov:2017,smith:2003}), insurance (\citet{embrechts:kluppelberg:mikosch:1997}), hydrology (\citet{anderson:meerschaert:1998}),  social networks (\citet{hofstad:2016,resnick:samorodnitsky:towsley:davis:willis:wan:2016}); see  \citet{clauset:shalizi:newman:2009, gabaix:2009} for a variety of other examples where heavy-tailed data appear. Considering such multivariate data, even if all marginals exhibit tail distributions heavier than normal or exponential, there is no reason to believe that they will all have the same or equivalent tail behavior. Classical approaches to modeling such data is built around assuming either equivalent tail behavior or a marginal transformation to equivalent tail behavior (\citet{resnickbook:2007,beirlant:goegebeur:teugels:segers:2004}); hence a modeling approach allowing univariate heavy-tailed variables to have different behavior across margins may be quite useful. On the matter of modeling dependence between margins, the Gaussian structure has a universal appeal due to its parsimonious and easily interpretable parameter set, as well as its generalizability to any dimension. In this article, a model for multivariate heavy-tailed data is proposed which allows the univariate tail distributions to be non-identical across different margins and the well-known Gaussian copula is used to characterize dependence.  Although Gaussian copula is widely known to have no tail dependence, using tools from extreme value theory we are able to  estimate not only all tail index parameters but also the Gaussian correlation parameters of this  model, attesting for the viability of Gaussian dependence as a modeling choice for multivariate heavy-tailed data with arbitrarily different marginal distributions.
%The results obtained adds to our understanding of modeling, estimation and tail dependence behavior of heavy-tailed data with Gaussian dependence.
%In this article, a model for multivariate heavy-tailed data is proposed which allows the univariate tail distributions to be non-homogeneous across different margins; moreover the well-known and popular Gaussian copula is used to characterize dependence.

Gaussian dependence with any choice of marginal distributions is described conveniently using \emph{copulas}. For a random vector $\bX=(X^1,\ldots, X^d)\sim F$ with  continuous marginal distributions $F_1,\ldots, F_d$, the associated copula $C:[0,1]^d \to [0,1]$ and survival copula $\widehat{C}:[0,1]^d \to [0,1]$ are defined as {functions} such that $F(\bx)=\P(\bX \le \bx) =C(F_1(x_1), \ldots, F_d(x_d)),$ and $\overline{F}(\bx)=\P(\bX > \bx) =\widehat{C}(\overline{F}_1(x_1), \ldots, \overline{F}_d(x_d)),$ respectively,  for $\bx= (x_1,\ldots,x_d)\in \R^d,$
where $\overline{F}_j=1-F_j\; \forall\, j\in\{1,\ldots,d\}.$ If $\Phi_\Sigma$ denotes a $d$-variate  normal distribution with all marginal means zero, variances one and positive semi-definite correlation matrix $\Sigma\in\R^{d\times d}$, and $\Phi$ denotes a standard normal distribution function, then $$C_{\Sigma}(u_1,\ldots,u_d)= \Phi_\Sigma(\Phi^{-1}(u_1),\ldots, \Phi^{-1}(u_d)), \quad \quad 0< u_1,\ldots, u_d < 1,$$
defines a Gaussian copula with correlation matrix $\Sigma$ with corresponding survival copula $\widehat{C}_{\Sigma}$. 

\subsection{The model}\label{subsec:model}
First we define the particular model of interest which incorporates  heavy-tailed marginals along with a Gaussian copula dependence structure.

\begin{definition}\label{def:rvgc}
An $\R^d$-valued random vector $\bX=(X^1,\ldots,X^d) \sim F$ follows a \emph{Pareto-tailed Gaussian copula distribution} with tail index parameters $\balpha=(\alpha_1,\ldots,\alpha_d) \in (0,\infty)^d$, scaling constant vector $\btheta=(\theta_1,\ldots,\theta_d)\in (0,\infty)^d$, and positive semi-definite correlation matrix $\Sigma$, if the following holds:
\begin{enumerate}[(i)]
    \item For all $j\in\mathbb{I}=\{1,\ldots, d\}$, marginal random variables $X^j\sim F_j$  where $F_j$ is continuous, strictly increasing and satisfies $\overline{F}_j(t):=1-F_j(t)\sim \theta_jt^{-\alpha_j}$ as $t\to\infty$ where $\alpha_j>0$, $\theta_j>0$.
     \item The joint distribution function $F$ of $\bX$ is given by  $$F(\bx) = C_{\Sigma}(F_1(x_1),\ldots,F_d(x_d)), \quad \bx=(x_1,\ldots,x_d)\in\R^d,$$
      where $C_{\Sigma}$ denotes the Gaussian copula with correlation matrix $\Sigma\in\R^{d\times d}$.
\end{enumerate}
We write in brief $\bX \in \PGC(\balpha, \btheta, \Sigma)$ or $F \in \PGC(\balpha, \btheta, \Sigma)$ where some parameters may be ignored for convenience. Alternatively, we call $\bX\in \PGC$ a \emph{Power-law-tailed Gaussian copula} random variable as well.
\end{definition}
For a bivariate $\PGC$ random vector $\bX$,  we  write $\bX\in \PGC(\alpha_1,\alpha_2,\theta_1,\theta_2,\rho)$ or $\bX\in\PGC(\alpha,\alpha_2,\rho)$, since $\Sigma$ is determined by one correlation parameter $\rho$. In this paper we often state the general result in $d$-dimensions but focus on 2-dimensional results for exposition purposes.

\begin{figure}[h]
    \centering
 \includegraphics[width=\linewidth]{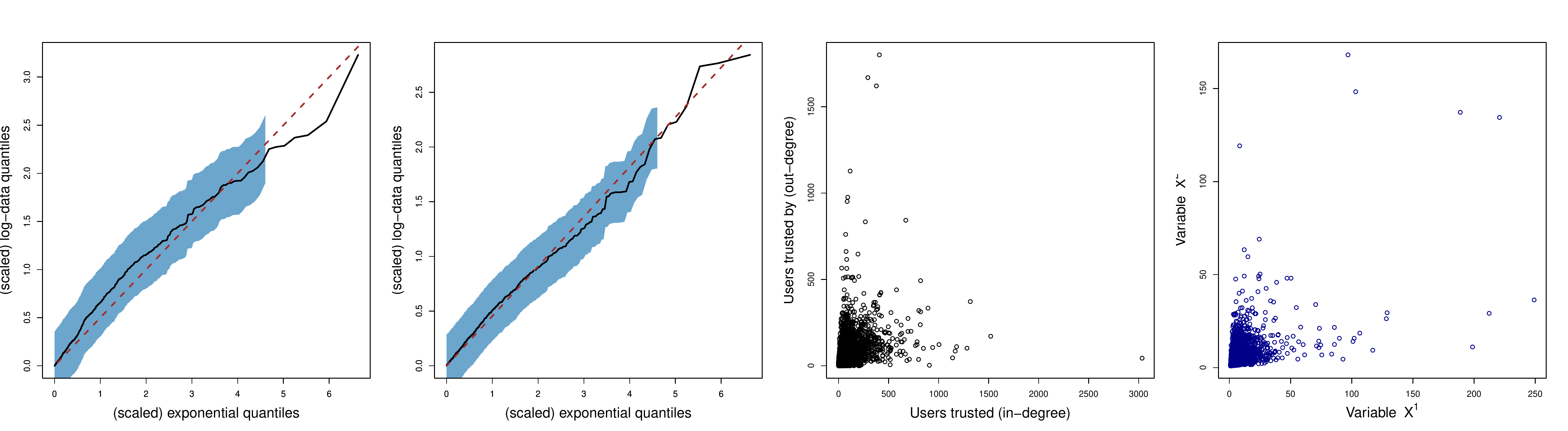} 
 \caption{Epinions user `trust' graph. Left two plots: exponential QQ plot of \underline{users trusted} (in-degree) and \underline{users trusted by} (out-degree) counts respectively; third plot: scatter plot of \underline{users trusted} vs \underline{users trusted by}; fourth plot: scatter plot of $n=75,879$ data points generated from $\PGC(\alpha_1=2, \alpha_2=2.2, \rho=0.8)$.} 
  \label{fig:epinions}
\end{figure}
To illustrate the merit of such a modeling paradigm we consider an online social network data from a  consumer reviews website \emph{Epinions.com}. The data is obtained from \href{https://snap.stanford.edu/}{https://snap.stanford.edu/} and  consists of a network where the nodes represent unique users, and, a directed edge between two nodes (users) indicates that the former user `trusts' the latter user (in the context of consumer reviews). The network has 75,879 nodes and  508,837 edges, and,  we form in-degrees and out-degrees for each node to investigate the global relationship between the two degree distributions. Now consider the bivariate data of in-degree and out-degree for each node. For data exploration purposes, we perform an exponential quantile-quantile plot (QQ plot) of both marginal distributions (the two left plots in \Cref{fig:epinions}). In an exponential QQ plot, quantiles of standard exponential distribution are plotted against logarithm of the sorted data (empirical quantiles); in this case it is done for top 1\% of the data (759 data points) in both plots.
The linearity of such exponential QQ plots, as observed in \Cref{fig:epinions} provides evidence of heavy-tailed behavior for both in- and out- degree distributions; additionally we have constructed 95\% confidence bands around them as well which contains the target line with appropriate slope; see \cite{das:resnick:2008, das:ghosh:2013} for details.  The third plot in \Cref{fig:epinions} is a scatter plot of the count of \emph{Users trusted} vs \emph{Users trusted by} for all the $n=75,879$ nodes (users), while the fourth plot is a scatter plot of $n=75,879$ data points generated randomly from $\PGC(\alpha_1=2, \alpha_2=2.2, \rho=0.8)$, where the parameter values are estimated from the \emph{Epinions} dataset using  techniques developed in this paper.  It is apparent from \Cref{fig:epinions} that the two scatter plots may be generated from the same or closely related multivariate distributions, and along with the two exponential QQ plots, this provides a reasonable premise for modeling such data using the Pareto-tailed Gaussian copula distributions.

A modest generalization of asymptotically Pareto-like or power-law functions, are \emph{regularly varying functions} which have been a useful paradigm for modeling heavy-tailed distributions. A measurable function $f: \R_{+} \to \R_{+}$ is \emph{regularly varying} (at $+\infty$) with some fixed $\beta\in \R$ if $$\lim_{t\to\infty} f(tx)/f(t) =x^{\beta}, \qquad \forall\,x>0.$$
    We write $f\in\RV_{\beta}$ and if $\beta=0$, we call $f$ \emph{slowly varying}. A real-valued random variable $X\sim F$  is regularly varying (at $+\infty$) if the tail 
  $\overline{F}:= 1-F \in \RV_{-\alpha}$ for some $\alpha> 0$.   
Popular examples of such distributions are Pareto, Fr\'ech\'et, Student's $t$; see \citet{embrechts:kluppelberg:mikosch:1997,resnickbook:2007} for details of theory and applications using regular variation.

\begin{remark} On the choice of Pareto-like or power-law tails as marginal distributions in the $\PGC$ model, a couple of remarks are due.
\begin{enumerate}[(a)]
   \item Although regularly varying distribution tails are a natural choice for modeling heavy-tails, for  various popular parametric distributions, like Pareto, Burr, Student's $t$, Fr\'echet, log-gamma, the tail distributions actually behave like a power-law, i.e., $\ov{F}(t) \sim \theta t^{-\alpha}$ for some $\alpha>0, \theta>0$, cf. \cite{embrechts:kluppelberg:mikosch:1997}. In \Cref{table:RVdist}, a list of such distributions is given, each of which also exhibit the so-called second order regular variation property (See \Cref{def:2rv}).
    \item The $\PGC$ distribution is also closely related to the  \emph{regularly varying tailed Gaussian copula distribution}, denoted $\RVGC(\alpha,b,\Sigma)$ as defined in \cite{das:fasen:2023a}. For $\bX\in\RVGC(\alpha,b,\Sigma)$, all marginal distributions are tail equivalent (having the same tail index parameter $\alpha$), which is a restriction from $\PGC$ distributions; on the other hand, the marginal distributions of RVGC random vectors are allowed to be regularly varying, which is a generalization from $\PGC$ distributions. 
\end{enumerate}
\end{remark}

Given an i.i.d sample from $F\in \PGC(\balpha,\btheta,\Sigma)$ our goal is to estimate the model parameters. Towards this end, first we provide asymptotic approximation for extreme tail probabilities of the form
\begin{align}\label{prob:tailset}
    \P(\bX\in tA) = \P(X^i > tx_i, i \in \mathbb{I}), \quad x_i>0, i \in \mathbb{I}
\end{align}
 where $\bX\sim F\in \PGC(\balpha,\btheta,\Sigma)$; this is established using tail asymptotics similar to the one used in \citep{das:fasen:2023a} for $\RVGC$  models. Now using the structure of the tail behavior, we propose consistent estimators for $\balpha$ and $\Sigma$, and finally prove asymptotic normality of the said estimators assuming a further second order regular variation condition on relevant univariate distribution tails.
In the next section, we discuss this particular model choice along with the dependence structure assumption  and briefly reflect on why estimation and inference for such a seemingly natural model has remained relatively unexplored.

\subsection{Motivation and related literature}\label{subsec:motivation}

Classically, heavy-tailed data have often been analyzed under the semi-parametric modeling paradigm of regular variation (\citet{resnickbook:2007}); here the marginal distributions are modeled separately as univariate heavy-tailed distributions whereas the dependence is captured using various parametric and non-parametric methods depending on the user's flavor, see  \cite{dehaan:ferreira:2006, beirlant:goegebeur:teugels:segers:2004,joe:1997} for various modeling paradigms. A traditional and widely adapted modeling choice in such cases is to transform all marginal distributions to have the same (or at least asymptotically equivalent) tail behavior first and then assess (tail) dependence separately \cite{jessen:mikosch:2006,resnickbook:2007,cai:einmahl:dehaan:2011}. Alternatively, one may also employ non-standard regular variation (\citet[Chapter 6.5.6]{resnickbook:2007}), or sparse regular variation (\citet{meyer:wintenberger:2021}), but  modeling and estimation pursuits along these lines are till date quite limited. In the $\PGC$ distribution framework, we refrain from standardizing the marginal tails; note that transforming marginal tails may have a discernible effect on joint tail dependence behavior as well, cf. \cite{das:fasen:2023a,furman:kuznetsov:su:zitikis:2016}.

%The matter of modeling (tail) dependence is understood easier with the notion of copulas. 
For characterizing tail dependence, a variety of tools exist including spectral measures (\citet{einmahl:dehaan:sinha:1997}), stable tail dependence functions (\citet{beirlant:goegebeur:teugels:segers:2004}), coefficients of tail dependence (\citet{ledford:tawn:1996}), tail copula (\citet{joe:1997}), etc. Although, under reasonable assumptions there exist estimates for tail risk probabilities using the aforementioned tools, but certain tail probabilities like \eqref{prob:tailset} are rendered negligible by the presence of a property called \emph{asymptotic tail independence}, which holds for a  random vector $\bY=(Y^1, Y^2)$ with identical marginals if $\P(Y^1>t|Y^2>t)\to 0$. This  means both variables $Y^1$ and $Y^2$ do not take high values together and it holds for the bivariate Gaussian copula \cite{sibuya:1960,ledford:tawn:1997}; see \cite{das:fasen:2023a} for further discussion on $d$-dimensional Gaussian copula. Incidentally, if we consider a bivariate power-law tailed model with unequal tail indices, classical regular variation produces an asymptotic limit model which has mass concentrated only on the co-ordinate axis of the heavier-tailed variable of the two, resulting again in a null estimate for joint tail sets (\cite{resnickbook:2007}). Hence such modeling is not helpful for parameter estimation, especially for estimating the correlation parameters.
%, or, equivalently, 
%\begin{align}\label{eq:asyind}
%\P(Y^1>t,Y^2>t) = o(\min_j  \P(Y^j>t)), \quad \text{ as } t\to \infty.
%\end{align}
%If $\widehat{C}$ is the survival copula of $\bY$, then \eqref{eq:asyind} is equivalent to $\widehat{C}(u,u)=o(u)$ as $u\to 0$. 
%For higher dimensional random vectors $\bY=(Y^1,\ldots,Y^d)\in \R^d$ with identical marginals, we may interpret asymptotic independence to be 
%\begin{align}\label{eq:asyindmulti}
%\P(Y^j>t, \forall j\in \mathbb{I}) = o\Big(\min_{\genfrac{}{}{0pt}{}{S \subseteq \mathbb{I}, }{0<|S|<d}} \P(Y^j>t, \forall j\in S)\Big).
%\end{align}
%The bivariate Gaussian copula possesses this property (see \citep{sibuya:1960,ledford:tawn:1997,fung:seneta:2011}), and this is {partially} exhibited for higher dimensional Gaussian copula in \citep{das:fasen:2023a}. Since explicit expression for the $o(\cdot)$ function on the right hand side of \eqref{eq:asyindmulti} are  difficult to compute, a null estimate is often the best available for $\P(\bY\in tA)$ as defined in \eqref{prob:tailset}; this is especially true for Gaussian copula with heavy-tailed margins. 
Referring to the scatter plots in \Cref{fig:epinions}, notice that although high values of both variables do not often occur together validating the said tail independence property, there are a few  observations where we find both variables taking relatively high values, although on a scale smaller than the individual high values. The techniques proposed in this paper are designed to estimate the correlation parameters from such behavior in the data, moreover this can be achieved in any general $d$-dimensions.
%In this paper we formalize the asymptotic rate with which this happens and use it for estimation. Note that such asymptotic probability rates for tail equivalent margins for such a model have already been characterized in \cite{das:fasen:2023a}.

From an applications  point of view, Gaussian copula has been used extensively in financial risk management starting from its broad introduction to the financial world by \citet{li:2000} to its notoriety as noted in \cite{salmon:2009}; see \cite{das:embrechts:fasen:2013,Donnelly:Embrechts} for further discussions. Gaussian copula also remains popular in  a wide variety of disciplines, including gene expression studies \cite{liu:fang:lafferty:wasserma:2012}, hydrology \cite{renard:2007}, economic forecasting \cite{wallach:2019}, just to mention a few. 
%On top of this, the empirical evidence of heavy-tails in various applications provides a natural incentive to analyze such a model and provide consistent estimates of model parameters if possible. 
To summarize, our interest in addressing this problem are the following:
\begin{enumerate}[(i)]
\item There is enough empirical evidence for multivariate heavy-tailed data and Gaussian copula is perhaps the most popular dependence structure; we also observe real data examples where such modeling seem appropriate (see \Cref{sec:realdata}).
    \item Although non-homogeneous heavy-tails seems to be a practical choice, joint estimation of unequal tail indices and tail dependence structures are not often addressed in the literature; in particular since most modeling paradigms assume marginal tail equivalence or a transformation to marginal tail equivalence.
     \item If data is limited, empirical estimation for tail probabilities are quite unreliable; hence a parametric model with appropriate tail parameter estimation may be helpful.
\end{enumerate}

%\subsection{Structure of the paper}\label{subsection:structure}
The paper is structured as follows.
We finish \Cref{sec:intro} with a table of notations which is useful throughout the paper. In \Cref{sec:probext} we provide explicit asymptotic expression for probabilities of  tail sets as defined in \eqref{prob:tailset}. The limit results from \Cref{sec:probext} allow us to propose estimators for the various parameters in the \emph{Pareto-tailed Gaussian copula model}  using ideas from extreme value theory. We estimate the tail indices and the correlation parameter in \Cref{sec:parmeterest} and show consistency of the proposed estimators. Asymptotic normality of the estimators are established in \Cref{sec:asynorm}. The efficacy of the estimation procedure is supported by an extensive simulation study in \Cref{sec:simdata}, here we generate data using a variety of heavy-tailed marginal distributions and Gaussian copula with different correlation parameters. Heavy-tailed data are abundant in various disciplines and in \Cref{sec:realdata} we fit the $\PGC$ model  on real data from online networks, insurance claims and internet traffic. Finally we conclude in \Cref{sec:concl} with discussion of the $\PGC$ model, its applicability and possible extensions. Some proofs and auxiliary results are relegated to Appendix \ref{sec:proofs} and \ref{sec:aux}.
\subsection{Notations}\label{subsec:notations}

A summary of relevant notations and concepts used in this paper are given here with references provided as required.
$$ 
\begin{array}{llll}
\RV(\beta) \text{ or } \RV_{\beta} & \text{Regularly varying functions with index $\beta\in\R$.}\\[2mm]
%We can and do}\\&\text{assume such functions are continuous and strictly increasing.}\\[2mm]
2\RV(\beta,\rho,A) & \text{Second order regularly varying functions with indices $\beta\in\R, \rho\le 0$ and}\\[2mm]
        & \text{scaling function $A$, cf. \Cref{def:2rv}. }\\[2mm]
\mathbb{I} & \text{The index set } \{1,\ldots,d\}.\\[2mm]
%\end{array}
%$$
% $$ 
%\begin{array}{llll}   
\bx\le \by & x_j\le y_j  \text{ for all  $j\in \mathbb{I}$  for vectors } \bx=(x_1,\ldots,x_d)^{\top}, \by=(y_1,\ldots,y_d)^{\top}.\\[2mm]
|S| & \text{Cardinality of the set } S.\\[2mm]
 \bzero, \bone, \binfty & \bzero=(0,\ldots,0)^\top, \bone=(1,\ldots,1)^\top, \binfty=(\infty,\ldots,\infty)^\top. \\[2mm]
  e_i & (0,\ldots,1,\ldots,0)^{\top}, i\in \mathbb{I} \text{ where $e_i$ has only one non-zero entry 1 at}\\[2mm] 
   & \text{the $j$-th co-ordinate.}\\[2mm]
\Sigma_{IJ} & \text{Appropriate sub-matrix of a given matrix } \Sigma\in\R^{d\times d} \text{ for non-empty } \\[2mm] 
            & I,J\subset \mathbb{I}; \text { also }\Sigma_I:=\Sigma_{II}. \quad \text{A similar notation is used for vectors} \\[2mm]
            & \text{like $\bone_I$ and $\bx_I$ as well for subsets } I\subset \mathbb{I}. \\[2mm]        
 z_{(1)} \geq\ldots\geq z_{(n)} \quad   &  \text{The decreasing order statistics of observations } z_1,\ldots,z_n\in \R.\\[2mm]  
% \mathds{1}_{A} & \text{Indicator function of the set } A.\\[2mm]
 tA & \{t\bx: \bx\in A\} \quad \text{for a constant } t>0 \text{ and } A\subset\R^d_+.\\[2mm]
% \varphi, \Phi & \text{The density  and distribution function, respectively of a standard }\\[2mm]
 %        & \text{normal distribution.}\\[2mm]
% \by_t=\bx_t+o(t) &  y_{t,j}-x_{t,j}=o(t) \text{ as }  t\to\infty \text{ for all  $j\in \mathbb{I}$  for sequences} \\[2mm]
% &  \bx_t=(x_{t,1},\ldots,x_{t,d}),\by_t=(y_{t,1},\ldots,y_{t,d})\in \R^d.\\[2mm]
 f(\bx) & \text{For any $f:\R\to\R$, $\bx\in\R^d$, define } f(\bx)=\left(f(x_1),\ldots, f(x_d)\right),  \\[2mm]
  & \text{ e.g., } \log{\bx}=(\log {x_1},\ldots,\log {x_d}),\,\,  \sqrt{\bx^{-1}}=\left(\frac1{\sqrt{x_1}},\ldots,\frac1{\sqrt{x_d}}\right).\\[2mm]
\end{array}
$$

\section{Probability of extreme tail sets}\label{sec:probext}
First we compute the probability of joint threshold crossings as defined in \eqref{prob:tailset} under a $\PGC$ distribution. This allows us to use tools from extreme value theory to estimate the model parameters. For multivariate normal random vectors, such joint threshold crossings have been studied in detail \citep{dai:mukherjea:2001,elnaggar:mukherjea:1999,hashorva:husler:2003,hashorva:2005,hashorva:2019}. In the context of heavy-tailed random vectors with Gaussian copula dependence, \cite{das:fasen:2023a} computes probabilities like \eqref{prob:tailset} and more general sets assuming the marginal distributions to be regularly varying and tail equivalent; moreover \cite{asmussen:rojas-nandayapa:2008} computes threshold crossings for sum of marginals which are lognormal. The result in this section partially extends \cite[Theorem 3.1]{das:fasen:2023a} to unequal tail indices.  The tail probability approximation given next requires the solution to the quadratic program $\mathcal{P}_{\Sigma^{-1},\sqrt{\balpha}}:  \min_{\{\bz\ge \sqrt{\balpha}\}} \bz^\top \Sigma^{-1}\bz$ for a positive definite matrix $\Sigma$ and $\balpha>\bzero$, which is given in \Cref{lem:qp}; see \Cref{subsec:proofprobext}.

\begin{theorem}\label{thm:mainrect}
    Let  an $\R^d$-valued random vector $\bX\in \PGC(\balpha, \btheta, \Sigma)$ where $\Sigma\in \R^{d\times d}$ is positive definite. Let $\gamma:=\gamma(\Sigma, \sqrt{\balpha})$, $I:=I(\Sigma,\sqrt{\balpha})$, $\bkappa:=\bkappa(\Sigma,\sqrt{\balpha})$, and $h_{i}:=h_i(\Sigma,\sqrt{\balpha}), i\in I$, be defined as  in the solution to the quadratic program $\mathcal{P}_{\Sigma^{-1},\sqrt{\balpha}}$, cf. \Cref{lem:qp}. Then for $\bx>\bzero$, as $t\to\infty$,
    \begin{align}\label{eq:mainlimitmrvA}
       \P(X^i>tx_i, \forall i\in \mathbb{I}) = (1+o(1))\Psi t^{-\gamma} (\log t)^{\frac{\Delta-|I|}{2}}\prod_{i\in I} x_i^{-\sqrt{\alpha_i}h_i},
    \end{align}
    where 
    \begin{align}
    \Delta & = {\sqrt{\balpha_{I}^{-1}}}^{\top}\Sigma_{I}^{-1}\sqrt{\balpha_{I}}, \label{def:Delta}\\
    \Psi  & =(4\pi)^{\frac{\Delta-|I|}{2}} |\Sigma_{I}|^{-1/2}\prod_{i\in I} \left[(\theta_i\sqrt{\alpha_i})^{\frac{h_i}{\sqrt{\alpha_i}}} h_i^{-1}\right] \times \P(\bY_J>\boldsymbol{l}_{\binfty}) \label{def:Psis}
    \end{align}
    with $\bl_{\infty}:=\lim_{t\to\infty}t(\sqrt{\balpha_J}-\bkappa_J)$, $J=\mathbb{I}\setminus I$ and $\bY_J\sim \mathcal{N}(\bzero_J, \Sigma_J-\Sigma_{JI}\Sigma_{I}^{-1}\Sigma_{IJ})$ if $J\neq \emptyset$ and $\P(\bY_J>\bl_{\binfty}):=1$ if $J=\emptyset$.
\end{theorem}
The proof of \Cref{thm:mainrect} is given in \Cref{subsec:proofprobext}.
\begin{remark}{\color{white}{space}}\\
\begin{enumerate}[(i)]
\item In this result we obtain tail asymptotics with unequal tail indices. In general, if we use standard multivariate regular variation, such a computation will provide trivial answers unless all marginals have equivalent tails.
    \item    \Cref{thm:mainrect} matches \cite[Theorem 3.1]{das:fasen:2023a} if  $\alpha_i$'s are same for all $i\in \mathbb{I}$ and $\theta_i$'s are the same for $i\in \mathbb{I}$. On the flip side, we are not able to show multivariate regular variation in cones of $\R_+^d$.
    \item If  we denote $Y=\min_{i\in \mathbb{I}} X^i$, then from \eqref{eq:mainlimitmrvA} we can check that as $t\to\infty$,
\[\P(Y>t) = \P(X^i> t, \forall i\in \mathbb{I}) \in \RV_{-\gamma}\]
where $\gamma=\gamma(\Sigma, \sqrt{\balpha})$. If all $\alpha=\alpha_i, \forall i \in \mathbb{I}$, then from \Cref{lem:qp}, we have \linebreak $\gamma=\alpha\bone^{\top}_I\Sigma_I^{-1}\bone_I$.
\end{enumerate}
 
\end{remark}

  It is instructive to formulate \Cref{thm:mainrect} in two dimensions where we can express the solution to the quadratic program $\mathcal{P}_{\Sigma^{-1},\sqrt{\balpha}}$ explicitly. This also allows us to observe the different asymptotic rates when the underlying parameters change. Note that in two dimensions positive definiteness of $\Sigma$ with one parameter $\rho$ is equivalent to $-1<\rho<1$.

\begin{corollary}\label{cor:2dim}
Let %$X_i\sim F_i$ where $\ov{F}_i(t)\sim \theta_it^{-\alpha_i}, i=1,2$ (as $t\to\infty$) with $\theta_1,\theta_2>0,\alpha_1\ge\alpha_2>0$ and 
$\bX\in\PGC(\alpha_1,\alpha_2,\theta_1,\theta_2,\rho)$  and $-1<\rho<1$. Then for any $x_1>0, x_2>0$, the following holds:
    \begin{enumerate}[(i)]
        \item If $\rho < \min(\sqrt{\alpha_2/\alpha_1},\sqrt{\alpha_1/\alpha_2})$, then
       \begin{align}
             &\P(X^1>tx_1, X^2>tx_2) \nonumber\\
               & \qquad\qquad = (1+o(1))\Psi t^{-\gamma} (\log t)^{-\frac{\rho}{2\sqrt{\alpha_1\alpha_2}}\gamma} x_1^{-\frac{\alpha_1-\rho\sqrt{\alpha_1\alpha_2}}{1-\rho^2}}x_2^{-\frac{\alpha_2-\rho\sqrt{\alpha_1\alpha_2}}{1-\rho^2}}, \quad t\to\infty \label{eq:cor:2dimmain}
       \end{align}
       where 
    \begin{align}
    \gamma & :=\frac{\alpha_1+\alpha_2-2\rho\sqrt{\alpha_1\alpha_2}}{(1-\rho^2)}, \label{def:pgcgamma}\\
    \Psi & :=\frac{{(4\pi)^{-\frac{\rho}{2\sqrt{\alpha_1\alpha_2}}\gamma}}{({1-\rho^2})^{3/2}}(\theta_1\sqrt{\alpha_1})^{\frac{\sqrt{\alpha_1}-\rho\sqrt{\alpha_2}}{1-\rho^2}}(\theta_2\sqrt{\alpha_2})^{\frac{\sqrt{\alpha_2}-\rho\sqrt{\alpha_1}}{1-\rho^2}}}{1-\rho(\alpha_1+\alpha_2)+\rho^2\sqrt{\alpha_1\alpha_2}}. \label{def:pgcpsi}
    \end{align}
    \item If $\rho = \min(\sqrt{\alpha_2/\alpha_1},\sqrt{\alpha_1/\alpha_2})$, then with $i^*:=\argmax_i \alpha_i,\, \gamma=\max(\alpha_1,\alpha_2):=\alpha_{i^*}$,
     \begin{align}
             \P(X^1>tx_1, X^2>tx_2) = (1+o(1)) \frac{\theta_{i^*}}{2} t^{-\gamma}x_{i^*}^{-\gamma}, \quad t\to \infty.
       \end{align}
         \item If $\rho > \min(\sqrt{\alpha_2/\alpha_1},\sqrt{\alpha_1/\alpha_2})$, then with $i^*:=\argmax_i \alpha_i,\, \gamma:=\max(\alpha_1,\alpha_2)=\alpha_{i^*}$,
     \begin{align}
             \P(X^1>tx_1, X^2>tx_2) = (1+o(1)) {\theta_{i^*}} t^{-\gamma}x_{i^*}^{-\gamma}, \quad t\to \infty.
       \end{align}
    \end{enumerate}
\end{corollary}
\begin{proof} 
The proof follows from \Cref{thm:mainrect} by appropriately computing the values from the solution of $\mathcal{P}_{\Sigma^{-1},\sqrt{\balpha}}$ in \Cref{lem:qp} where  $\balpha=(\alpha_1,\alpha_2)$. 
\begin{enumerate}[(i)]
    \item  If $\rho < \min(\sqrt{\alpha_2/\alpha_1},\sqrt{\alpha_1/\alpha_2})$, then the solution to $\mathcal{P}_{\Sigma^{-1},\sqrt{\balpha}}$ is given by
    \[\gamma= \balpha^{\top}\Sigma^{-1}\balpha= \frac{\alpha_1+\alpha_2-2\rho\sqrt{\alpha_1\alpha_2}}{(1-\rho^2)}.\]
    Moreover $I=\{1,2\}, J=\emptyset, h_1 = \frac{\sqrt{\alpha_1}-\rho\sqrt{\alpha_2}}{1-\rho^2}, h_2 = \frac{\sqrt{\alpha_2}-\rho\sqrt{\alpha_1}}{1-\rho^2},  \Delta= \frac{2\sqrt{\alpha_1\alpha_2}-\rho(\alpha_1+\alpha_2)}{(1-\rho^2)\sqrt{\alpha_1\alpha_2}}$, and $|\Sigma|=1-\rho^2$. Using these values in \eqref{eq:mainlimitmrvA} lead to the result.

    \item Let $\rho = \min(\sqrt{\alpha_2/\alpha_1},\sqrt{\alpha_1/\alpha_2})$. W.l.o.g. assume $\alpha_1>\alpha_2$, hence $\rho=\sqrt{\alpha_2/\alpha_1}$.  Then $I=\{1\},\,J=\{2\}$. Thus $\gamma=\alpha_1, h_1= \sqrt{\alpha_1}, \Delta=1, |\Sigma_I|=1$.  Also $\kappa_2=\rho\sqrt{\alpha_1}=\sqrt{\alpha_2}$ and hence  $l_{\infty}=0$ implying $\P(Y_2>l_{\infty}) = \P(Y_2>0)=\frac12$ where $Y_2\sim \mathcal{N}(0,1-\rho^2)$. Hence the result follows.
    \item  Again, w.l.o.g. assume $\alpha_1>\alpha_2$. Then the case for $\rho >  \min(\sqrt{\alpha_2/\alpha_1},\sqrt{\alpha_1/\alpha_2}) = \sqrt{\alpha_2/\alpha_1}$ is the same as the previous one except that $\kappa_2=\rho\sqrt{\alpha_1}>\sqrt{\alpha_2}$ and hence  $l_{\infty}=-\infty$ implying $\P(Y_2>l_{\infty}) = \P(Y_2>-\infty)=1$, leading to the result.
\end{enumerate}
\end{proof}

\begin{remark}\label{rem:2dimprobext}
    \Cref{cor:2dim} leads to a few insights on the behavior of probabilities of extreme tail sets. %We have assumed $\alpha_1\ge \alpha_2>0$ without loss of generality.
    \begin{enumerate}[(a)]
        \item If $\rho< \min(\sqrt{\alpha_2/\alpha_1},\sqrt{\alpha_1/\alpha_2})$, then $Y=\min(X^1,X^2)\in \RV_{-\gamma}$ where $\gamma$ is defined in \eqref{def:pgcgamma}. Hence, if we can estimate $\alpha_1,\alpha_2, \gamma$ using sample data, we may estimate $\rho$ as a solution to a quadratic equation in $\rho$ defined by  \eqref{def:pgcgamma}. Furthermore, since this can be done for any 2-dimensional subset of a $d$-dimensional $\bX\in \PGC(\balpha,\Sigma)$ where $\Sigma=((\rho_{ij}))_{i,j\in \mathbb{I}}$, we can similarly estimate $\rho_{ij}$ for any $i\neq j \in \mathbb{I}$ if $\rho_{ij}<\min(\sqrt{\alpha_j/\alpha_i},\sqrt{\alpha_i/\alpha_j})$. We employ this idea for parameter estimation in \Cref{sec:parmeterest}.
        \item Depending on the value of the correlation parameter $\rho\in (-1,1)$, the value of $\gamma$ has some restrictions. We have
\begin{align*}
    \bullet\; \; \gamma & > \alpha_1+\alpha_2,  && \text{if } \rho<0,\\
    \bullet\; \; \gamma & =\alpha_1 +\alpha_2, && \text{if } \rho=0,\\
    \bullet\; \; \gamma & \in \left(\max(\alpha_1,\alpha_2),\alpha_1+\alpha_2\right), & &\text{if } \rho \in \Big(0, \min\Big(\sqrt{\alpha_2/\alpha_1},\sqrt{\alpha_1/\alpha_2}\Big)\Big),\\
    \bullet\; \; \gamma &=\max(\alpha_1,\alpha_2),  && \text{if } \rho \ge \min\Big(\sqrt{\alpha_2/\alpha_1},\sqrt{\alpha_1/\alpha_2}\Big).
\end{align*}
  These can be used to check appropriateness of a $\PGC$ assumption for sample data as well.       
        %For any value of $\rho\in (-1,1)$, $\gamma\ge \max(\alpha_1,\alpha_2)$. Moreover,
        %\begin{itemize}
       %  \item if $\rho<0$ then $\gamma>\alpha_1+\alpha_2$,
       %     \item if $\rho=0$ then  $\gamma=\alpha_1+\alpha_2$,
      %      \item if $0<\rho<\min(\sqrt{\alpha_2/\alpha_1},\sqrt{\alpha_1/\alpha_2})$ then $\gamma<\alpha_1+\alpha_2$,
      %      \item if $\rho\ge \min(\sqrt{\alpha_2/\alpha_1},\sqrt{\alpha_1/\alpha_2})$ then %$\gamma=\max(\alpha_1,\alpha_2)$.
    %    \end{itemize}
       
    \end{enumerate}
\end{remark}

\section{Parameter estimation}\label{sec:parmeterest}
Suppose $\bX_1,\ldots, \bX_n$ are i.i.d. from $F\in \PGC(\balpha,\btheta,\Sigma)$ and we are interested in estimating the model parameters, specifically $\balpha$ and $\Sigma$. In \Cref{sec:parmeterest} we found asymptotic probabilities of tail sets for the $\PGC$ model and the results indicated regular variation behavior for the tail distribution of a few random elements, which we make use of to estimate the parameters. For estimating the tail index of a regularly varying tailed distribution we use the celebrated  estimator of \citet{hill:1975}. If $Z_1,\ldots, Z_n$ are i.i.d. $F$ where $\ov{F}\in \RV_{-\alpha}, \alpha>0$, then the Hill estimator is defined as
\[ H_{k,n} = \frac1k \sum_{i=1}^k \log\left(\frac{Z_{(i)}}{Z_{(k+1)}}\right), \quad 1\le k <n \]
where $Z_{(1)}\ge \ldots Z_{(n)}$ denotes the decreasing order statistics of $Z_1,\ldots, Z_n$. It is well known $H_{k,n} \stp \alpha^{-1}$ as $n\to\infty, k=k(n)\to \infty$ with $k(n)/n\to 0$ \citep{hill:1975,mason:1982}. In fact Mason's result \cite[Theorem 2]{mason:1982} indicates that  the convergence of Hill's estimator is equivalent to regular variation of the underlying distribution tail making this estimator an ideal choice for testing the regular variation property as well. Here we note that for our next result any consistent estimator of the tail index parameter suffices, e.g., Pickand's estimator \citep{pickands:1975}, Dekkers-Einmahl-de Haan estimator \citep{dekkers:einmahl:dehaan:1989}, etc. We state the result in 2-dimensions; an extension to $d$-dimensions is evident and is implied by the 2-dimensional result, see \Cref{thm:multidimconv} for an explicit formulation.

\begin{theorem}\label{thm:2dimconv}
    Suppose $\bX_1,\ldots,\bX_n$ are i.i.d. bivariate random vectors where $\bX_i=(X^1_{i},X^2_{i})\in \PGC(\alpha_1,\alpha_2, \theta_1, \theta_2, \rho)$ for some $\alpha_1,\alpha_2,\theta_1,\theta_2>0 $ and $\rho\in (-1,1)$. 
   Let $Y_i=\min(X^{1}_i,X^{2}_i) \sim G, i=1,\ldots n$ . Define the following (Hill) estimators for $k_1=k_1(n), k_2=k_2(n), k_{12}=k_{12}(n) $ where $1 \le  k_1,k_2,k_{12} < n$:
 %  $\gamma=\frac{\alpha_1+\alpha_2-2\rho\sqrt{\alpha_1\alpha_2}}{1-\rho^2}$
   \begin{align*}
       \hat{\alpha}_j^{-1} &:= H^{\alpha_j}_{k_j} = \frac{1}{k_j} \sum_{i=1}^{k_j} \log\left(\frac{X^j_{(i)}}{X^j_{(k_j+1)}}\right),j=1,2, \\%\label{eq:hillbeta}\\
       \hat{\gamma}^{-1} &:= H^{\gamma}_{k_{12}} = \frac{1}{k_{12}} \sum_{i=1}^{k_{12}} \log\left(\frac{Y_{(i)}}{Y_{(k_{12}+1)}}\right). %\label{eq:hillgamma}
   \end{align*}
 Then the following holds:
    \begin{enumerate}[(i)]
       \item We have
       \[ \ov{G} \in \RV_{-\gamma} \quad \text{where} \quad \gamma = \begin{cases}
           \frac{\alpha_1+\alpha_2-2\rho\sqrt{\alpha_1\alpha_2}}{1-\rho^2} & \text{ if } \rho<\min\big(\sqrt{\alpha_2/\alpha_1},\sqrt{\alpha_1/\alpha_2}\big),\\
           \max(\alpha_1,\alpha_2) & \text{otherwise}.
       \end{cases}\]
       \item As $n\to \infty, k_j \to \infty$ with $k_j/n\to 0$ for $j=1,2,12$, we have 
       \begin{align}\label{eq:cinpwithgamma}
     \left(H^{\alpha_1}_{k_1},H^{\alpha_2}_{k_2}, H^{\gamma}_{k_{12}}\right) \stp  \left(\frac{1}{\alpha_1},\frac{1}{\alpha_2}, \frac{1}{\gamma}\right).
       \end{align}
       \item Define \begin{align}\label{def:rhohat}
           \hat{\rho} := \frac{\sqrt{\hat{\alpha}_1\hat{\alpha}_2}-\sqrt{\hat{\alpha}_1\hat{\alpha}_2+\hat{\gamma}^2-\hat{\gamma}(\hat{\alpha}_1+\hat{\alpha}_2)}}{\hat{\gamma}}.
       \end{align} If $\rho<\min\big(\sqrt{\alpha_2/\alpha_1},\sqrt{\alpha_1/\alpha_2}\big)$ then  as $n\to \infty, k_j \to \infty$ with $k_j/n\to 0$ for $j=1,2,12$, we have
       \begin{align}\label{eq:cinpwithrho}
       \left(\hat{\alpha}_1,\hat{\alpha}_2,\hat{\rho}\right) \stp (\alpha_1, \alpha_2, \rho).
       \end{align}
   \end{enumerate}
\end{theorem}

\begin{proof} {\color{white}{yes}}
    \begin{enumerate}[(i)]
        \item From \Cref{cor:2dim} (i), for $\rho<\min\big(\sqrt{\alpha_2/\alpha_1},\sqrt{\alpha_1/\alpha_2}\big)$, using $x_1=x_2=1$ in \eqref{eq:cor:2dimmain}, we get 
        \[\P(Y_1>t) =\P(X_1^1>t,X_1^2>t) = (1+o(1)) \Psi t^{-\gamma} (\log t)^{-\frac{\rho}{\sqrt{\alpha_1\alpha_2}}\gamma}\]
 where $\Psi>0$ and $\gamma=\frac{\alpha_1+\alpha_2-2\rho\sqrt{\alpha_1\alpha_2}}{1-\rho^2}$. Hence $\ov{G} \in \RV_{-\gamma}$. Similarly from \Cref{cor:2dim} (ii) and (iii)  for $\rho\ge\min\big(\sqrt{\alpha_2/\alpha_1},\sqrt{\alpha_1/\alpha_2}\big)$, we have $\ov{G} \in \RV_{-\gamma}$ where $\gamma=\max(\alpha_1,\alpha_2)$.
\item Let $X_1^j\sim F_j, j=1, 2$. Then by definition, $\ov{F}_j\in \RV_{-\alpha_j}, j=1,2$, moreover, in part (i) we have shown that $G\in\RV_{-\gamma}$. Clearly, $\hat{\alpha}_j^{-1}, \hat{\gamma}^{-1}$ are respectively Hill estimators for $\alpha_j^{-1}, j=1,2$ and $\gamma$, cf. \cite{hill:1975}. Hence by consistency of the Hill estimator (\cite{hill:1975,mason:1982}), we have consistency (in probability) of the individual estimators. Since all random variables are defined on the same probability space and the estimators converge to fixed parameter values, the joint consistency follows (using for example \cite[Proposition 3.1]{resnickbook:2007}).

\item If $\rho<\min\big(\sqrt{\alpha_2/\alpha_1},\sqrt{\alpha_1/\alpha_2}\big)$, then we have consistent estimates for $\alpha_1, \alpha_2$ and \linebreak $\gamma=\frac{\alpha_1+\alpha_2-2\rho\sqrt{\alpha_1\alpha_2}}{1-\rho^2}$. Hence to obtain an estimate of $\rho$, we solve the quadratic equation:
\[ g(\rho)= \hg\rho^2-2\sqrt{\hat{\alpha}_1\hat{\alpha}_2}\rho+(\hat{\alpha}_1+\hat{\alpha}_2-\hg)=0.\]
By \Cref{lem:solverho}, the unique feasible solution to $g(\rho)=0$ is given by 
\[\hat{\rho} = \frac{\sqrt{\hat{\alpha}_1\hat{\alpha}_2}-\sqrt{\hat{\alpha}_1\hat{\alpha}_2+\hat{\gamma}^2-\hat{\gamma}(\hat{\alpha}_1+\hat{\alpha}_2)}}{\hat{\gamma}}. \]
Now \eqref{eq:cinpwithrho} follows by using a continuous mapping theorem (\cite{billingsley:1968}) on \eqref{eq:cinpwithgamma} for \linebreak  $\rho<\min\big(\sqrt{\alpha_2/\alpha_1},\sqrt{\alpha_1/\alpha_2}\big)$.
       \end{enumerate}
\end{proof}

\subsection{Parameter estimation in higher dimensions}\label{sec:highdim}

The consistency of the estimators as exhibited in \Cref{thm:2dimconv} can be easily extended from two dimensions to higher dimensions. Since the correlations are computed from the two relevant margins, this is a natural extension from the bivariate result.

\begin{theorem}[Consistency]\label{thm:multidimconv}
    Suppose $\bX_1,\ldots,\bX_n$ are i.i.d. random vectors in $\R^d$ where $\bX_i=(X^1_{i},\ldots, X^d_{i})\in \PGC(\balpha, \btheta, \Sigma)$ for some $\balpha>0, \btheta>0$ and positive definite $\Sigma$. Let $\Sigma = ((\rho_{jl}))$ where $\rho_{jl}\in (-1,1)$ if $j\neq l$ and $\rho_{jj}=1, j\in \mathbb{I}$.
 Define $Y^{jl}_i=\min(X_i^j,X_i^l), i=1,\ldots n$ for $j,l\in \mathbb{I}, j\neq l$ and 
   $$\gamma_{jl}=\begin{cases} \frac{\alpha_j+\alpha_l-2\rho_{jl}\sqrt{\alpha_j\alpha_l}}{1-\rho_{jl}^2} & \text{if } \rho<\min \big(\sqrt{\alpha_l/\alpha_j},\sqrt{\alpha_j/\alpha_l}\big),\\
   \max (\alpha_j, \alpha_l) & \text{otherwise}. \end{cases}$$ Also define the following (Hill) estimators for $k=k(n)$  with $1 < k < n$:
   \begin{align*}
         \hat{\alpha}_j^{-1} &:= H^{\alpha_j}_{k,n} = \frac{1}{k} \sum_{1=1}^{k} \log\left(\frac{X^j_{(i)}}{X^j_{(k+1)}}\right),j\in \mathbb{I},\\
        \hat{\gamma}_{jl}^{-1} &:= H^{\gamma_{jl}}_{k,n} =  \frac{1}{k} \sum_{1=1}^{k} \log\left(\frac{Y^{jl}_{(i)}}{Y^{jl}_{(k+1)}}\right), \quad j,l\in \mathbb{I}, j\neq l.
   \end{align*}

  %  If $\rho_{jl}<\min\big(\sqrt{{\alpha_l}/{\alpha_j}},\sqrt{{\alpha_j}/{\alpha_l}}\big)$ for all $j\neq l$ 
    Then the following holds:
    \begin{enumerate}[(i)]
       \item  For $Y_1^{jl}\sim G_{jl}$, we have $\ov{G}_{jl}\in \RV_{-\gamma_{jl}}$.
       \item As $n\to \infty, k_n\to \infty$ with $k_n/n\to 0$, we have $\hat{\alpha}_j^{(k_n)}\stp\alpha_j, j=1,\ldots,n$ and $\hat{\gamma}^{(k_n)}_{jl} \stp \gamma_{jl}$.
       \item Define $\hat{\rho}_{jl}^{(k_n)} = \frac{\sqrt{\hat{\alpha}_j\hat{\alpha}_l}-\sqrt{\hat{\alpha}_j\hat{\alpha}_l+\hat{\gamma}_{jl}^2-\hat{\gamma}(\hat{\alpha}_j+\hat{\alpha}_l)}}{\hat{\gamma}_{jl}}, j\neq l$ and $\hat{\rho}_{jj}^{(k_n)}=1$. Let $\hat{\Sigma} =((\hat{\rho}_{jl}))$.  Then if $\rho_{jl}<\min\big(\sqrt{{\alpha_l}/{\alpha_j}},\sqrt{{\alpha_j}/{\alpha_l}}\big)$ for all $j\neq l$ as $n\to \infty, k_n\to \infty$ with $k_n/n\to 0$, we have
       \[ \left(\hat{\balpha},\hat{\Sigma}\right) \stp (\balpha, \Sigma).\]
   \end{enumerate}
\end{theorem}

\begin{proof}
    The proof follows from \Cref{thm:2dimconv} and is omitted here.
\end{proof}

%\begin{theorem}[Asymptotic Normality]
%\end{theorem}

\begin{remark}[On estimating $\btheta$]
Let $\bX_1,\ldots, \bX_n$ be i.i.d. $\PGC(\balpha,\btheta,\Sigma)$. In this paper we concentrate on understanding tail behavior and dependence, and hence our focus is on estimating $\balpha$ and $\Sigma$. Nevertheless, a consistent estimator of $\btheta$ can easily be proposed.  For $i=1,\ldots, d$, if $\hat{\alpha}_i$ is a consistent estimator of $\alpha_i$, then
\begin{align}\label{eq:estC}
    \hat{\theta}_i = \frac{k}{n} (X^i_{(k)})^{\hat{\alpha}_i}
\end{align}
is a consistent estimator of $\theta_i$, i.e., $\hat{\theta}_i\stp \theta_i$ as $n\to\infty, k\to \infty, k/n \to 0$; cf. \citet{hill:1975} for the case where $\alpha_i$ is estimated by the Hill estimator.
\end{remark}

\section{Asymptotic normality of the estimators}\label{sec:asynorm}

In \Cref{sec:parmeterest} we used the Hill estimator for estimating the tail index parameters. In order to provide confidence bounds around the estimators we need a weak consistency property for the estimators. It is well-known that the Hill estimator is asymptotically normal with a non-random constant mean under a \emph{second order regularly varying} assumption on the relevant distributional tail  (\citep{dehaan:resnick:1996}). Interestingly, second order regular variation, with appropriate formulation turns out to be equivalent to asymptotic normality of Hill's estimator. We use \citet[Theorem 4]{stupfler:2019} to show asymptotic normality of the multivariate Hill estimator under the $\PGC$ model. In addition to  second order regular variation on the marginals, the result also requires an appropriate condition on the joint distribution tail which is satisfied by the $\PGC$ distribution,  cf. \Cref{prop:condR}.

\begin{definition}[Second order regular variation; cf. \cite{geluk:dehaan:resnick:starica:1997, dehaan:ferreira:2006}]\label{def:2rv}
    A function $U$ on $[0,\infty)$ is second-order regularly varying (at $+\infty$) with first-order parameter $\beta$ and second-order parameter $\rho\le 0$ if there exists a function $A(t)\to 0$ (of ultimately constant sign) such that:
    \begin{align*}
        \lim_{t\to\infty} \frac{\frac{U(tx)}{U(t)} - x^{\beta}}{A(t)} = x^{\beta}\frac{x^{\rho}-1}{\rho}.
    \end{align*}
    We write $U \in 2\RV(\beta,\rho, A)$.
\end{definition}
%\begin{remark}\label{rem:2rvUequivF}
 Let $Z\sim F$ and define its quantile function $U(t)=\ov{F}^{\leftarrow}(1-1/t)$. For the asymptotic normality of multivariate Hill estimators, \citet{stupfler:2019} assumes each marginal quantile $U \in 2\RV$. 
Incidentally, for some $\alpha>0$ and $\rho\le 0$,  $\ov{F}\in 2\RV(-\alpha,\rho, A)$ is equivalent to $U\in 2\RV(1/\alpha, \rho/\alpha, A)$, cf. \cite[Theorem 2.3.9]{dehaan:ferreira:2006} or \cite[Lemma 3.2]{das:kratz:2020}. Hence we assume $2\RV$ on the relevant distribution tails for our result. In \Cref{table:RVdist} we provide a list of marginal distributions $F$ where $\ov{F}(x)\sim C x^{-\alpha}$ (as $x\to\infty$) and $\ov{F} \in 2\RV(-\alpha,\rho,A)$. Incidentally, although the standard Pareto distribution has exact power law tails, it fails to satisfy $2\RV$; nevertheless the Hill estimator still remains asymptotically normal in this case (\cite{hall:1982}).
%\end{remark}

We need a further result on bivariate tail dependence for obtaining the asymptotic normality of the estimators of the tail indices; this  is stated next. For any jointly distributed vector $(Z^1,Z^2) $ with $Z^i\sim F_{Z^i},\,i=1,2$, the tail dependence function $R_{Z^1,Z^2}$ on $[\bzero,\binfty]\setminus\{\binfty\}$ is defined as
\[R_{Z^1,Z^2}(z_1,z_2) :=\lim_{t\to\infty} t\, \P\left(\ov{F}_{Z^1}(Z^1)\le \frac {z_1}t, \ov{F}_{Z^2}(Z^2)\le \frac {z_2}t\right),\]
Incidentally, $R_{Z^1,Z^2}$ also uniquely characterizes the bivariate stable tail dependence function (\cite{drees:huang:1998}), see \cite{stupfler:2019,cai:einmahl:dehaan:zhou:2015}. The next result provides limits for this quantity for $\PGC$ distributions.

\begin{proposition}\label{prop:condR}
  Let $\bX=(X^1,X^2)\in \PGC(\alpha_1,\alpha_2,\theta_1,\theta_2,\rho)$ where $\alpha_1,\alpha_2,\theta_1,\theta_2>0$ and $\rho\in (-1,1)$. Also assume that $\rho<\min(\sqrt{\alpha_2/\alpha_1},\sqrt{\alpha_1/\alpha_2})$.
  Denote $X^i\sim F_i, i=1,2$ and $Y=\min(X^1,X^2)\sim G$. Then the following holds for $(w_1,w_2)\in [0,\infty)^2$:
  \begin{align}
    % & R_{1,2}(w_1,w_2)  :=\lim_{t\to\infty} t\, \P\left(\ov{F}_1(X^1)\le \frac {w_1}t, \ov{F}_2(X^2)\le \frac {w_2}t\right) =0 \label{eq:RX1X2}.
    & R_{1,2}(w_1,w_2)  := R_{X^1,X^2} (w_1,w_2) =0 \label{eq:RX1X2},\\
  %\end{align}
  %Moreover,
  %\begin{align}
     & R_{1,12}(w_1,w_2) := R_{X^1,Y} (w_1,w_2) =\begin{cases} 0, & \rho < \sqrt{\alpha_2/\gamma} \\
                                                   \frac{w_1}{2}, & \rho = \sqrt{\alpha_2/\gamma}\\
                                                   w_1, & \rho > \sqrt{\alpha_2/\gamma} \end{cases}, \label{eq:RX1Y}\\
     & R_{2,12}(w_1,w_2) :=R_{X^2,Y} (w_1,w_2) =\begin{cases} 0, & \rho < \sqrt{\alpha_1/\gamma} \\
                                                \frac{w_2}{2}, & \rho = \sqrt{\alpha_1/\gamma}\\
                                                 w_2, & \rho > \sqrt{\alpha_1/\gamma} \end{cases}. \label{eq:RX2Y}
  \end{align}  
  %\begin{align}
   %  & R_{1,12}(w_1,w_2) :=\lim_{t\to\infty} t\, \P\left(\ov{F}_1(X^1)\le \frac {w_1}t, \ov{G}(Y)\le \frac {w_2}t\right) =\begin{cases} 0 & \rho < \sqrt{\alpha_2/\gamma} \\
    %                                              \theta_1w_1^{-1}/2 & \rho = \sqrt{\alpha_2/\gamma}\\
     %                                              \theta_1 w_1^{-1} & \rho > \sqrt{\alpha_2/\gamma} \end{cases}, \label{eq:RX1Y}\\
    % & R_{2,12}(w_1,w_2) :=\lim_{t\to\infty} t\, \P\left(\ov{F}_2(X^2)\le \frac {w_1}t, \ov{G}(Y)\le \frac {w_2}t\right) =\begin{cases} 0 & \rho < \sqrt{\alpha_2/\gamma} \\
    %                                              \theta_2w_2^{-1}/2 & \rho = \sqrt{\alpha_1/\gamma}\\
    %                                               \theta_2 w_2^{-1} & \rho > \sqrt{\alpha_1/\gamma} \end{cases}. \label{eq:RX2Y}
 % \end{align}  
\end{proposition}

 The proof of \Cref{prop:condR} is given in \Cref{subsec:proofcondR}.

\begin{theorem}\label{thm:mainasynormbiv}
      Suppose $\bX_1,\ldots,\bX_n$ are i.i.d. bivariate random vectors with $\bX_i=(X^1_{i},X^2_{i})\sim F \in \PGC(\alpha_1,\alpha_2, \theta_1, \theta_2, \rho)$, with $X_i\sim F_i, i=1,2 $ and the following holds.
      \begin{enumerate}[(i)]
      \item Let $\rho< \min(\sqrt{\alpha_2/\alpha_1},\sqrt{\alpha_1/\alpha_2})$ and let $\gamma:= \frac{\alpha_1+\alpha_2-2\rho\sqrt{\alpha_1\alpha_2}}{1-\rho^2}$.
      \item Define $Y_i=\min(X^{1}_i,X^{2}_i)\sim G, i=1,\ldots n$, and define (Hill) estimators $ H^{\alpha_1}_{k_1}, H^{\alpha_2}_{k_2}, H^{\gamma}_{k_{12}}$ as in \Cref{thm:2dimconv}. 
      \item Let $\ov{F}_i\in 2\RV(-\alpha_i,\delta_i, A_i), i=1,2, \delta_1,\delta_2 \le 0$. From \Cref{thm:2dimconv} we have $\ov{G}\in\RV(-\gamma)$; additionally assume $\ov{G}\in 2\RV(-\gamma,\delta_{12}, A_{12})$ for $\delta_{12}\le 0$.
      \end{enumerate}
For $k_1=k_1(n), k_2=k_2(n), k_{12}=k_{12}(n) $ where $1 < k_1,k_2,k_{12} < n$, assume that as $n\to \infty, k_j\to \infty$ with $k_j/n\to 0, j=1,2,12$ we have
      \begin{itemize}
          \item $\sqrt{k_j} A_j(n/k_j) \to \lambda_j \in \R$,
          \item $k_1/k_j\to q_j\in (0,\infty)$.
      \end{itemize}
      Then as $n\to \infty$,
   \begin{align}\label{eq:2dimnorm}
       \sqrt{k_1} \left(\left(H^{\alpha_1}_{k_1},H^{\alpha_2}_{k_2}, H^{\gamma}_{k_{12}}\right) - \left(\frac{1}{\alpha_1},\frac{1}{\alpha_2}, \frac{1}{\gamma}\right)\right) \std \mathcal{N}(\mu(\bq), \Gamma(\bq))
   \end{align}
   where
   \[ \mu_j(\bq) = \sqrt{q_j}\frac{\lambda_j}{1-\delta_{j}}, \quad\quad \Gamma_{j,k}(\bq) = \begin{bmatrix} \frac{1}{\alpha_1^2} & \frac{R_{1,2}(q_1,q_{2})}{\alpha_1\alpha_2} & \frac{R_{1,12}(q_1,q_{12})}{\alpha_1\gamma}\\[1em]  \frac{R_{1,2}(q_1,q_{2})}{\alpha_1\alpha_2} &\frac{1}{\alpha_2^2} & \frac{R_{2,12}(q_2,q_{12})}{\alpha_2\gamma}\\[1em] \frac{R_{1,12}(q_1,q_{12})}{\alpha_1\gamma} & \frac{R_{2,12}(q_2,q_{12})}{\alpha_2\gamma} & \frac{1}{\gamma^2} \end{bmatrix}.\]
and $R_{1,2}, R_{1,12}, R_{2,12}$ are as defined in \Cref{prop:condR}.
  \end{theorem}
 \begin{proof}
      The proof is a direct application of \cite[Theorem 4]{stupfler:2019}, since all the necessary conditions are satisfied.
  \end{proof}
In \Cref{thm:mainasynormbiv}, the bias term $\mu(\bq)$ in the limit is not necessarily zero (component-wise) and depends on the second order regularly varying parameter of the relevant tail distribution (cf. \cite{stupfler:2019}); nevertheless, with an appropriate choice of $k_j$, it is possible to have $\lambda_j=0, j=1,2,12$ making $\mu(\bq)=\bzero$. In real data examples, such choices need to be made depending on data diagnostics. The appropriate choice of the intermediate sequence $k_j$ still remains relatively unresolved even in univariate tail estimation problems, see \cite{nguyen:samorodnitsky:2012,clauset:shalizi:newman:2009} for some proposed options.
%\begin{remark}
%    If $\rho<\min(\sqrt{\alpha_1/\gamma},\sqrt{\alpha_2/\gamma})$ then the limit covariance matrix $\Gamma$ turns out to be a diagonal matrix.
%\end{remark}

  \begin{corollary}\label{cor:onrho}
   Let the assumptions and notations of \Cref{thm:mainasynormbiv} hold. Then, as $n\to \infty, k_j\to \infty$ with $k_j/n\to 0, j=1,2,12$, so that $\lambda_1=\lambda_2=\lambda_{12}=0$, we have
  \begin{align}\label{eq:rho}
       \sqrt{k}_1\left(\hat{\rho} -\rho\right) \std \mathcal{N}(0, \nu)
   \end{align}
   where $\nu= \left(\nabla h(\alpha_1,\alpha_2,\gamma)\right)^{\top} \Gamma(\bzero) (\nabla h(\alpha_1,\alpha_2,\gamma))$ with $\Gamma(\bzero)=\text{Diag}\left(1/\alpha_1^2,1/\alpha_2^2,1/\gamma^2\right)$  and
   $$\rho=h(\alpha_1,\alpha_2,\gamma) = \frac{\sqrt{{\alpha}_1{\alpha}_2}-\sqrt{{\alpha}_1{\alpha}_2+{\gamma}^2-{\gamma}({\alpha}_1+{\alpha}_2)}}{{\gamma}}.$$   
  \end{corollary}

  \begin{proof}
     The value of $\Gamma(\bzero)$ is clear from \Cref{prop:condR} and \Cref{thm:mainasynormbiv}. The result then follows by using the Delta method on \eqref{eq:2dimnorm}. 
  \end{proof}

\begin{table}[ht]
\begin{center}
\renewcommand{\arraystretch}{2}
    \begin{tabular}{|l|l|l|c|c|} \hline
       \sc{Distribution} & Parameters & CDF $F(x)$ or pdf $f(x)$ & $\alpha$ & $\rho$ \\\hline\hline
        Burr & $\beta, \sigma>0$ & $F(x)=1- (1+x^{\sigma})^{-\beta}, \, x>0$ & $\beta\sigma$ & $-\beta$ \\ [1ex]\hline 
        Fr\'echet & $\beta,\sigma>0, \mu\in \R$  & $F(x)=e^{-\left(\frac{x-\mu}{\sigma}\right)^{-\beta}}, \, x>\mu$ & $\beta$ & $-\beta$\\ [1ex]\hline 
        Generalized Pareto (GPD) & $\xi,\sigma>0, \mu\in \R$ & $F(x) = 1- \left(1+\frac{\xi(x-\mu)}{\sigma}\right)^{-\frac1{\xi}}, x\ge\mu$ & $\frac{1}{\xi}$ & $-1$\\ [1ex]\hline 
         Hall/Weiss & $\beta, \sigma>0$ & $F(x)=1- \frac12x^{-\beta}(1+x^{-\sigma}), \, x>0$ & $\beta$ & $-\sigma$ \\ [1ex]\hline 
        Inverse-gamma & $\beta, \sigma>0$ & $F(x)=1- \frac{\Gamma(\beta,\sigma/x)}{\Gamma{(\beta)}}, \, x>0$& $\beta$ & $-1$\\ [1ex]\hline 
        Log-gamma$^*$ & $\beta>0$ & $F(x)=1- x^{-\beta}(1+\log x), \, x>0$& $\beta$ & $0$\\ [1ex]\hline 
        Student's $t$ & $\nu >0$ & $f(x)=\frac{\Gamma((\nu+1)/2)}{\sqrt{\pi\nu}\Gamma(\nu/2)}\left(1+\frac{x^2}{\nu}\right)^{-\frac{\nu+1}{2}}, \, x\in \R$ & $\nu$ & $-2$\\ [1ex]\hline 
    \end{tabular} 
    \renewcommand{\arraystretch}{1}
\end{center}
    \caption{Univariate distribution functions $F$ where $\ov{F}(x)\sim C x^{-\alpha}$ (as $x\to\infty$) and $\ov{F} \in 2\RV(-\alpha,\rho,A)$. Here $\Gamma(s)$ and $\Gamma(s,y)$ denote the Gamma function and Incomplete Gamma integral respectively, cf. \cite{gradshteyn2014table}. The definition of Log-gamma distribution mentioned is  non-standard and hence appears as Log-gamma$^*$.} \label{table:RVdist}
\end{table}

\begin{remark}
  
A multivariate extension for the asymptotic normality as indicated in \Cref{thm:mainasynormbiv} and \Cref{cor:onrho} is possible, but requires careful bookkeeping of notations, especially while computing covariances. Individual asymptotic normality of parameters are evident from \Cref{thm:mainasynormbiv}.

\end{remark}

\section{Simulation study}\label{sec:simdata}
We exhibit weak consistency of the proposed estimators using a variety of simulated data. The codes used for computing the estimates and generating the plots are available at \href{https://github.com/bikram-jit-das/ParetoGaussCop}{https://github.com/bikram-jit-das/ParetoGaussCop}.
\begin{figure}[ht]
    \centering
    \includegraphics[width=\linewidth]{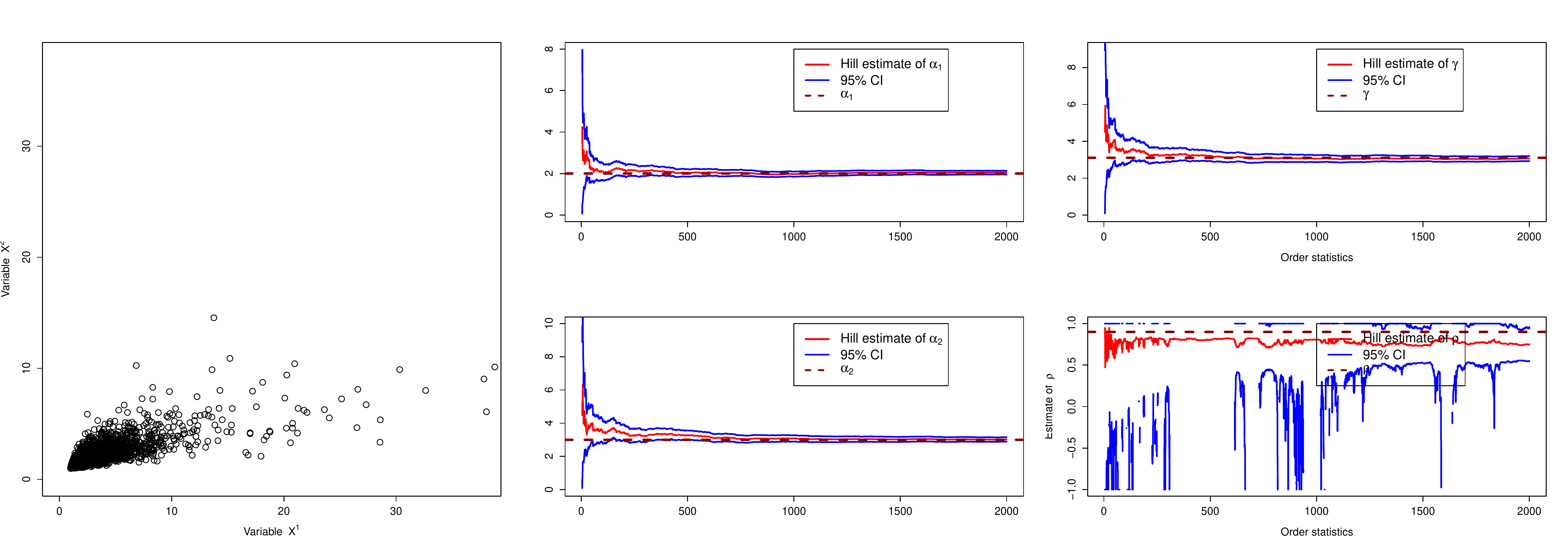} \\ 
 \includegraphics[width=\linewidth]{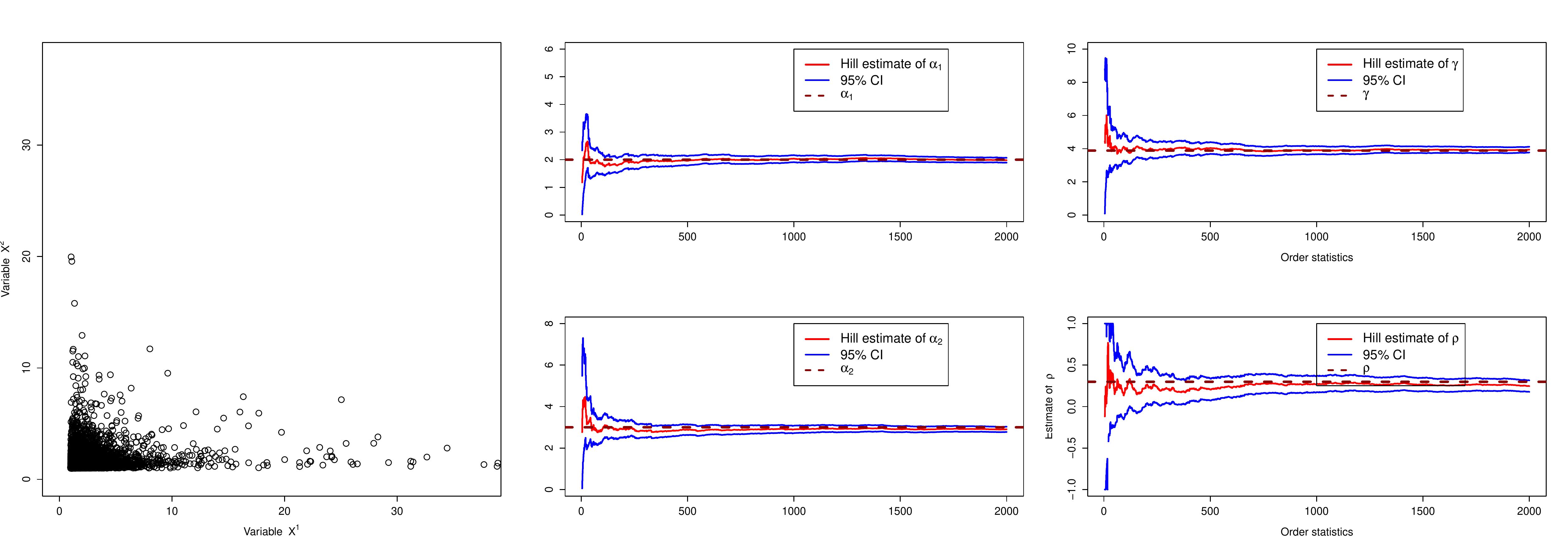} \\ 
  \includegraphics[width=\linewidth]{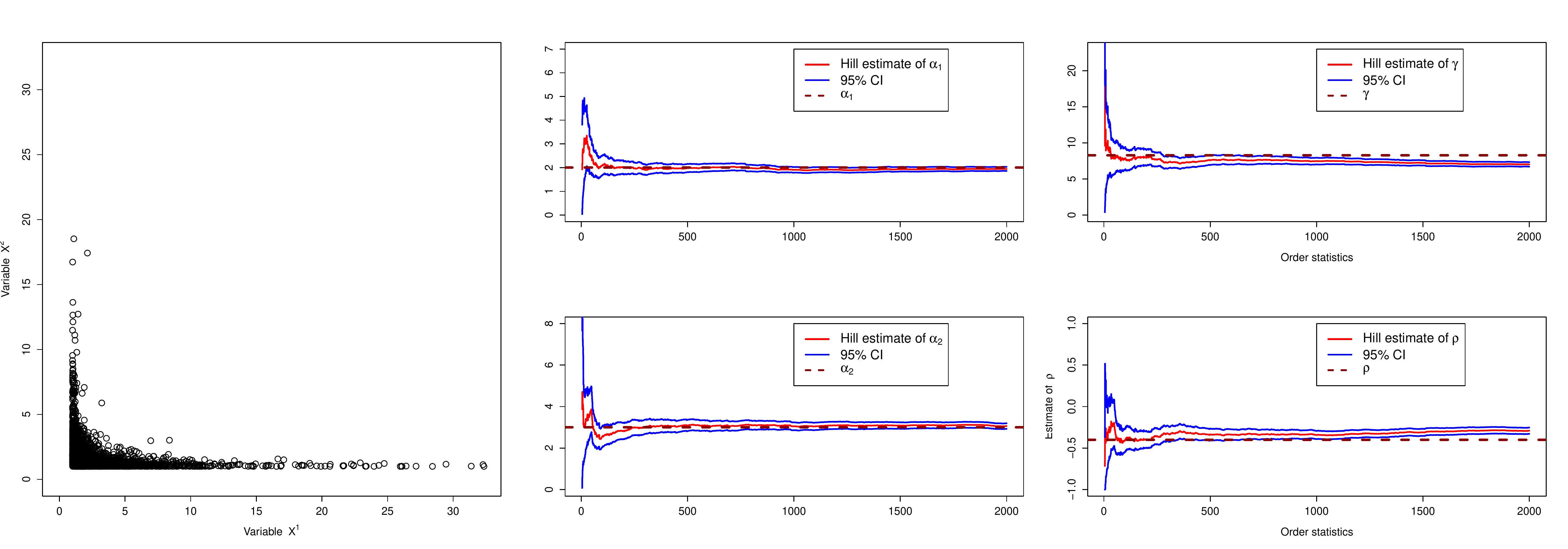}  
 \caption{Parameter estimates from $\PGC$ distributions with Pareto(2) and Pareto(3) marginals from $n=10000$ simulations. (1) Top row: $\rho = 0.9$, (2) middle row: $\rho =  0.3$, (3) bottom row: $\rho= - 0.4$. In each row, (a) left plot: scatter plot of data, (b) middle plot: Hill estimates of $\alpha_1$ and $\alpha_2$ for $20< k \le 2000$ with 95\% confidence bands, (c) right plot: estimates of $\gamma$ and $\rho$ for $20< k \le 2000$ with 95\% confidence bands.The dashed lines in the same color indicate the real parameter value.} 
  \label{fig:pareto}
\end{figure}

We generate $n=10000$ data points from bivariate  $F\in \PGC(\alpha_1,\alpha_2, \theta_1=1, \theta_2=1, \rho)$ with the following choices:
\begin{enumerate}[(1)]
    \item $F_i\sim \text{Pareto} (\alpha_i), i=1,2$ where $\alpha_1=2, \alpha_2=3$.  We take three choices for the correlation parameter $\rho= 0.9, 0.3, -0.4$. Note that since $F_i$'s are exactly Pareto, $\ov{F}_i \not\in 2\RV$. Although we cannot apply \Cref{thm:mainasynormbiv} to get confidence bounds, we know that Hill estimator for $\alpha_i$ here are asymptotically normal (\cite{hall:1982}). We still use the formula from \Cref{thm:mainasynormbiv} to create confidence bounds for $\hat{\gamma}$ and $\hat{\rho}$ here along with those of $\hat{\alpha}_1$ and $\hat{\alpha}_2$.
    \item $F_i\sim \text{Fr\'echet} (\alpha_i,\mu=0,\sigma=1), i=1,2$. Here $\ov{F}_i \in 2\RV(-\alpha_i,-1)$, hence we can and do produce confidence bounds for this case using \Cref{thm:mainasynormbiv}. We take three choices for the correlation parameter $\rho= 0.8, -0.3, -0.8$.

\end{enumerate}

\begin{figure}[ht]
    \centering
 \includegraphics[width=\linewidth]{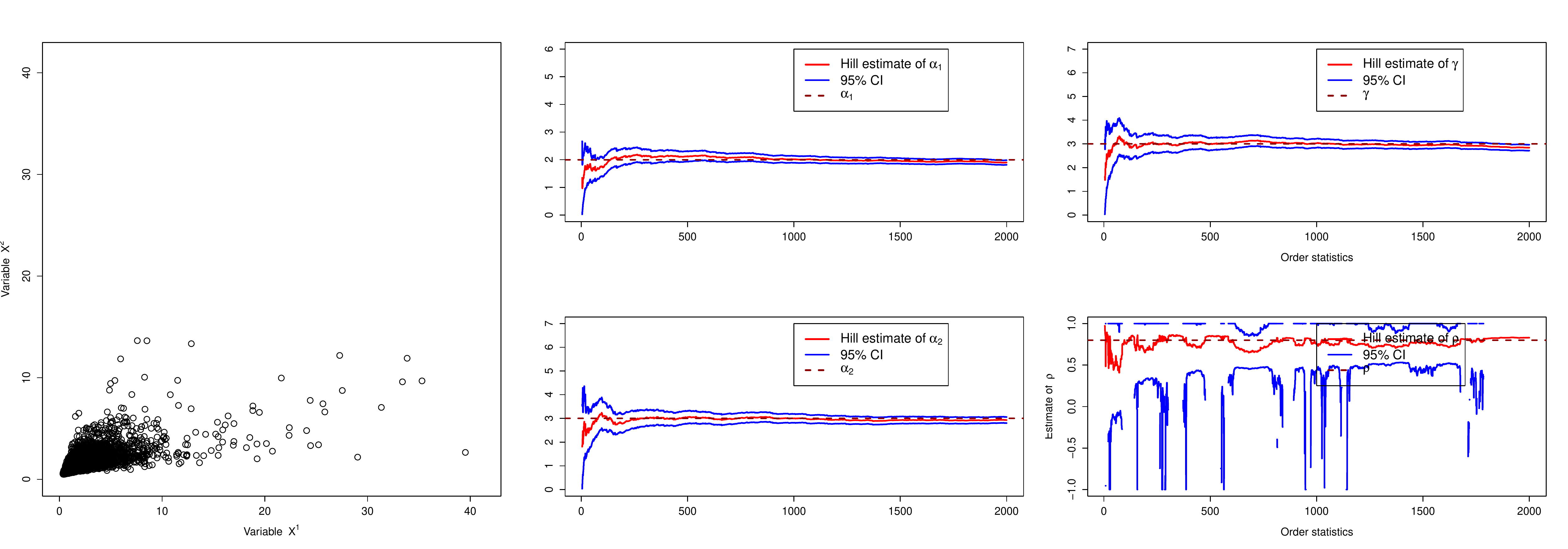} \\ 
  \includegraphics[width=\linewidth]{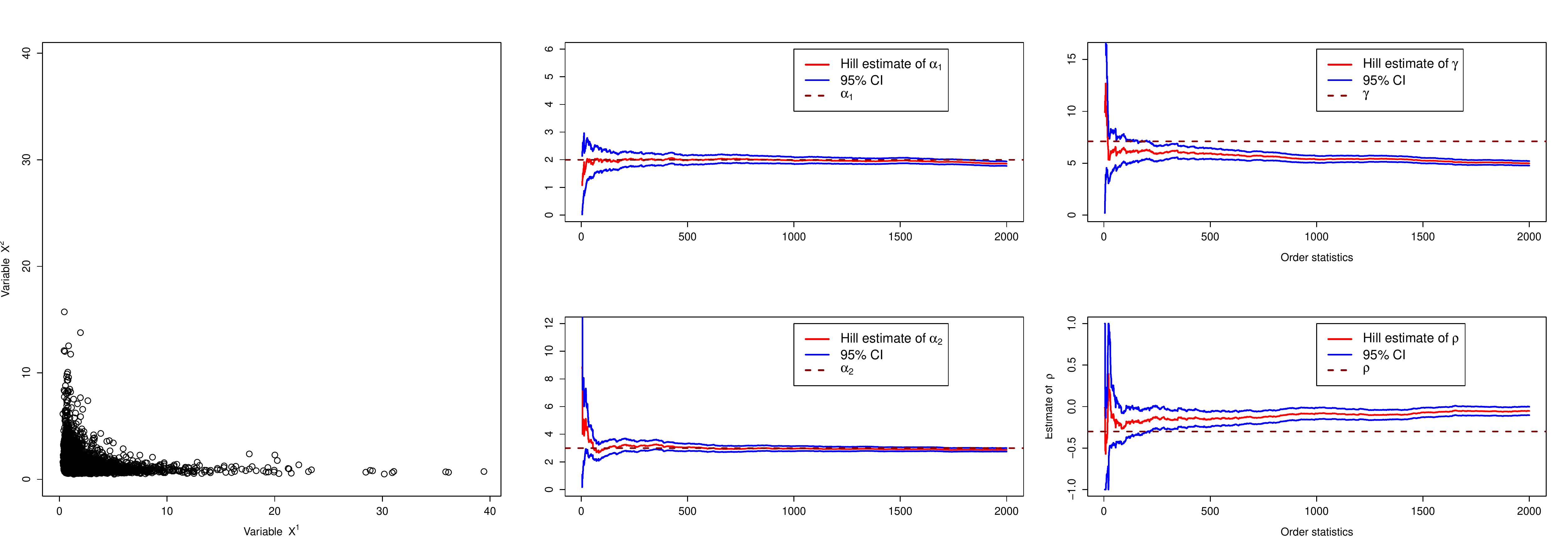}  \\
  \includegraphics[width=\linewidth]{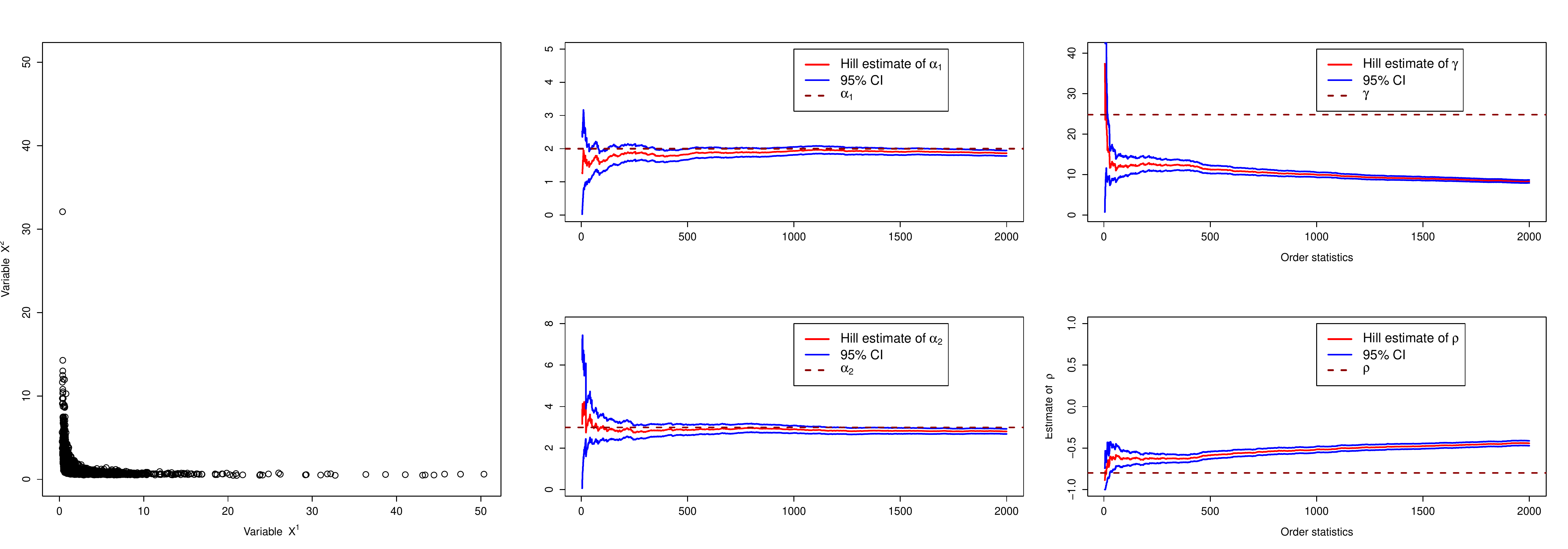} 
   \caption{Parameter estimates from $\PGC$ distributions with Fr\'echet (2,1) and Fr\'echet(3,1) marginals from $n=10000$ simulations. (1) Top row: $\rho = 0.8$, (2) middle row: $\rho =  - 0.3$, (3) bottom row: $\rho= - 0.8$. In each row, (a) left plot: scatter plot of data, (b) middle plot: Hill estimates of $\alpha_1$ and $\alpha_2$ for $20< k \le 2000$ with 95\% confidence bands, (c) right plot: estimates of $\gamma$ and $\rho$ for $20< k \le 2000$ with 95\% confidence bands.The dashed lines in the same color indicate the real parameter value.} 
  \label{fig:frechet}
\end{figure}

The Hill estimators defined in \Cref{thm:2dimconv} are used to estimate $\alpha_1, \alpha_2$ and $\gamma$.
In both \Cref{fig:pareto} and \Cref{fig:frechet}, we have plotted the scatter plot, Hill estimators for the parameters $\alpha_1, \alpha_2$ and $\gamma$, as well as the estimator for $\rho$ based on the Hill estimators of $\alpha_1, \alpha_2$ and $\gamma$. Note that the estimate for $\rho$ is valid if $\rho<\min(\sqrt{\alpha_2/\alpha_1},\sqrt{\alpha_1/\alpha_2}) = 0.82$ in both the Pareto and Fr\'echet cases described. Both the Hill estimates and estimates of $\rho$ are plotted for top $k=1, \ldots,2000$ order statistics and a stability in the plots indicates their consistency  \citep{resnickbook:2007}. 

From \Cref{fig:pareto}, it is clear that $\alpha_1, \alpha_2$ are well-estimated. The estimate of $\gamma$ works well for  positive $\rho$ values, but has a little bias for $\rho<0$. For the $\rho$ values, the estimates are reasonable, although for $\rho= 0.9$ there is a negative bias, moreover the confidence bounds are erratic. Note that since $\rho=0.9>\sqrt{\alpha_1/\alpha_2}=\sqrt{2/3}=0.82$, \Cref{thm:2dimconv} does not propose a consistent estimate for $\rho$ in this case and confidence bounds are even far-fetched; nevertheless our estimates are indicative of a high value for the correlation parameter.

Similarly in \Cref{fig:frechet}, the marginal parameters $\alpha_1, \alpha_2$ are well-estimated. For $\rho= 0.8$, the estimates for $\gamma$ and $\rho$ seem to perform reasonably well, although the confidence bounds for the estimate of $\rho$ behave a bit erratically possibly due to their closeness to the boundary of validity $\min(\sqrt{\alpha_2/\alpha_1},\sqrt{\alpha_1/\alpha_2}) = 0.82$.  For negative values of $\rho$, the estimates seem to exhibit a significant bias, but they do seem to identify the significant negative dependence.

The scatter plots of the data also reveal a little more about the structure of dependence in the data. In both Figures \ref{fig:pareto} and \ref{fig:frechet}, when $\rho>0$, the positive association is evident (although theoretically the model has asymptotic tail independence), on the other hand when $\rho<0$, the scatter plot hug the axes (which is classical evidence for asymptotic independence), but the amount of negative dependence is not quite clear. Nevertheless, the estimates of $\rho$ do identify the values to be positive or negative, albeit with some bias in case $\rho<0$.

\section{Data analysis}\label{sec:realdata}
In this section we investigate four different data sets where diagnostics (e.g., exponential QQ plot \cite{das:resnick:2008,kratz:resnick:1996}) indicate that the marginal distributions are heavy-tailed. Hence a $\PGC$ model is a valid choice and we fit this model. We use the Hill estimator to estimate the tail index parameters and estimate the Gaussian correlation parameter using the estimator for $\rho$ defined in \eqref{def:rhohat}. We plot each estimator for the top few quantiles and stability of the estimators indicate both justification for heavy-tail assumption (\cite{mason:1982}) and a choice of quantile for a point estimate. We also provide 95\% confidence intervals for each estimator using the parameter estimates to compute the variances according to \Cref{thm:mainasynormbiv}.

\begin{example}[Online network data]\label{ex:network}
We analyze two  data sets from the popular repository of network data \href{https://snap.stanford.edu/}{https://snap.stanford.edu/} maintained by Stanford Network Analysis Platform.
\begin{figure}[ht]
    \centering
    \includegraphics[width=\linewidth]{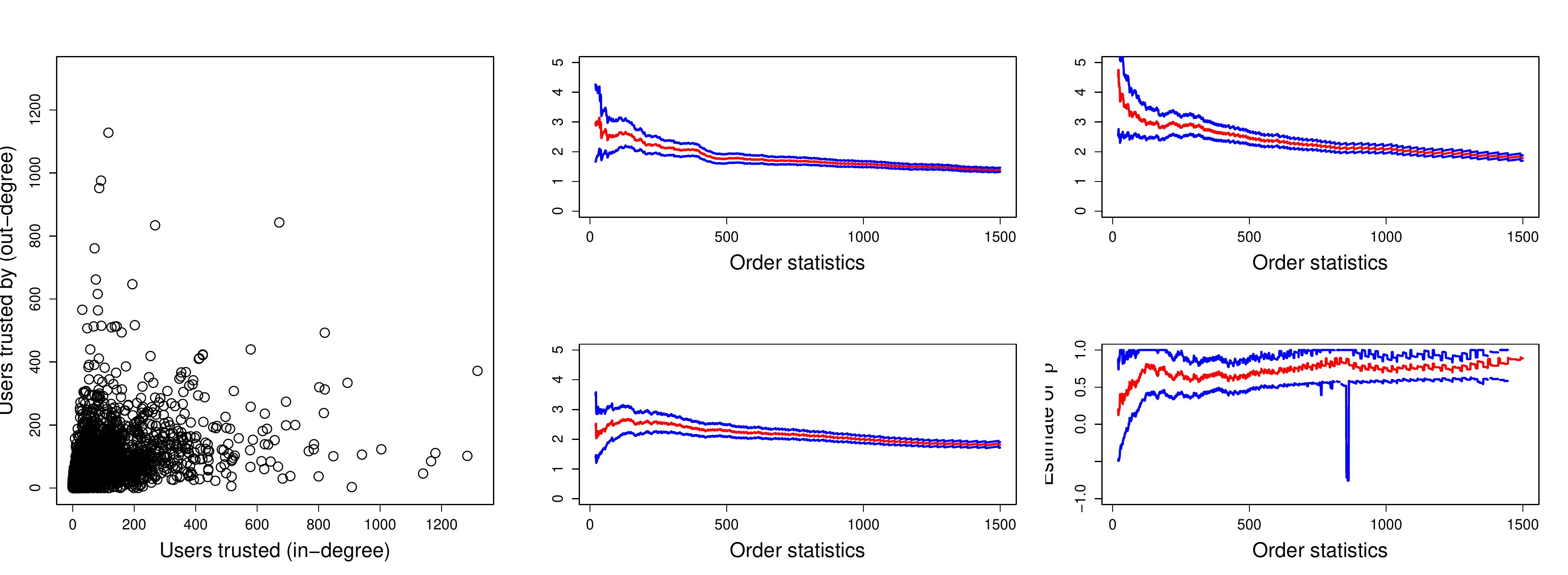} 
 \caption{Epinions online network data. (i) Left plot: scatter plot of \underline{users trusted} (in-degree) vs \underline{users trusted by} (out-degree), (ii) middle plot: Hill estimates of tail indices of \underline{users trusted} (top) and \underline{users trusted} by (bottom), (iii) right plot: Hill estimates of tail index of $\min$(in-degree, out-degree) at the top and estimate of $\rho$ (bottom). Plots of estimates are for top 1500 quantiles and with 95\% estimated confidence intervals.} 
  \label{fig:degsocep}
\end{figure}

\begin{enumerate}[(1)]
    \item The first dataset is from \emph{Epinions.com} which was described in \Cref{sec:intro} for the diagnostic plots in \Cref{fig:epinions}, where we could infer that the in-degrees (users trusted) and out-degrees (users trusted by) appear heavy-tailed using exponential QQ plots. In \Cref{fig:degsocep}, we plot the Hill estimates for the tail parameters $\alpha_1, \alpha_2$ and $\gamma $ for the data of in-degree, out-degree and $\min$(in-degree, out-degree) respectively; we also estimate the Gaussian correlation parameter $\rho$ using our estimate from \eqref{def:rhohat}.  Estimates for  ${\alpha}_1$ and $\alpha_2$ seem to be close to 2 whereas those for $\gamma$ seem close to 3. The estimates for $\rho$ appear to be between $(0.6,0.8)$ indicating a high positive association. Hence we may infer users who are trusted by more users also know more users they can trust. 
    %This dataset was also investigated in \citet{richardson_etal:2003}.

    \item The second data set is a citation graph from e-print \href{http://arxiv.org}{arXiv} fo high energy physics phenomenology  publications. It covers all the citations within a dataset of 34,546 papers with 421,578 edges (from January 1993 to April 2003). If a paper $A$ cites paper $B$, the graph contains a directed edge from $A$ to $B$; external citations are ignored, see \href{https://www.cs.cornell.edu/projects/kddcup/}{2003 KDD Cup}, where the data was first released. 
   
\begin{figure}[ht]
    \centering
    \includegraphics[width=\linewidth]{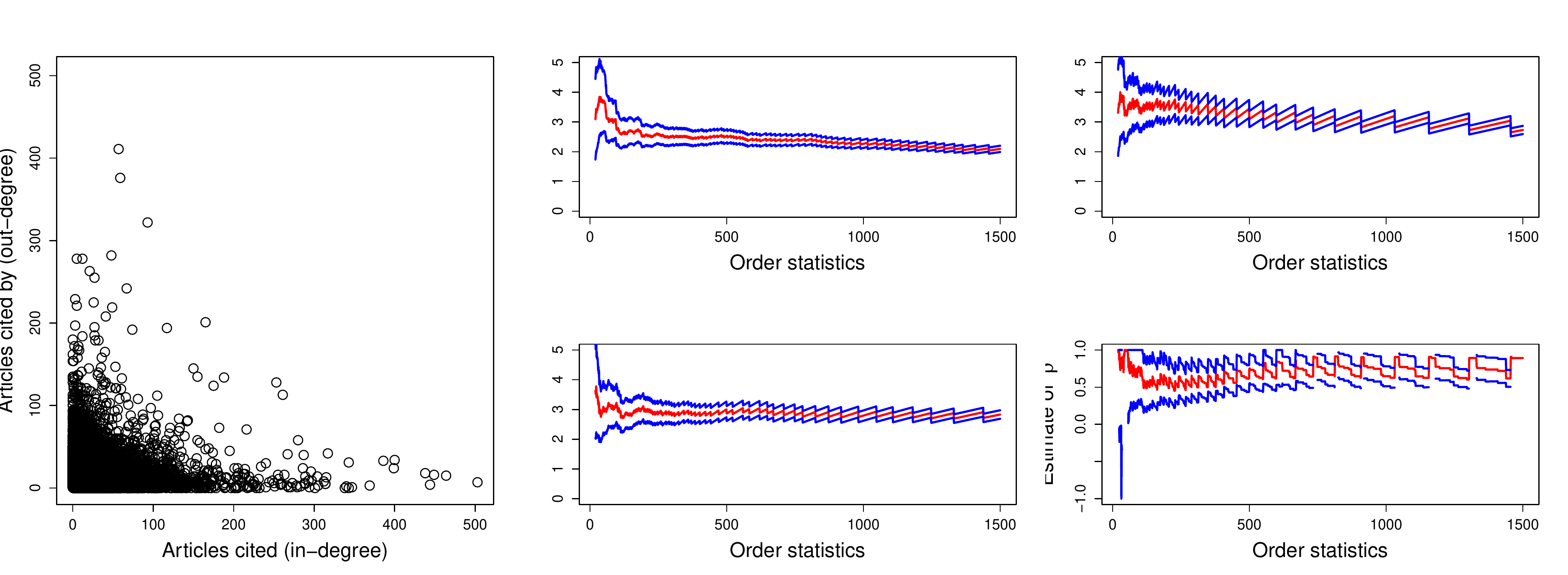} 
 \caption{High energy Physics citation data from \href{http://arxiv.org}{arXiv}. (i) Left plot: scatter plot of \underline{articles cited} (in-degree) vs \underline{articles cited by} (out-degree), (ii) middle plot: Hill estimates of tail indices of \underline{articles cited} (top) and \underline{articles cited by} (bottom), (iii) right plot: Hill estimates of tail index of $\min$(in-degree, out-degree) at the top and estimate of $\rho$ (bottom). Plots of estimates are for top 1500 quantiles and with 95\% estimated confidence intervals.} 
  \label{fig:CiteHp}
\end{figure} 
    Similar to the \emph{Epinions} network data we compute in-degree and out-degree for each node (paper) and then compute Hill estimates of the tail indices $\alpha_1$ (in-degree), $\alpha_2$ (out-degree) and $\gamma$ ($\min$(in-degree, out-degree)), and additionally $\rho$. They are all plotted in \Cref{fig:CiteHp}.  Here $\hat{\alpha}_1$ and $\hat{\alpha_2}$ seem to be between 2 and 3 whereas $\hat{\gamma}$ is between 3 and 4. The estimate of $\rho$ is above 0.5 indicating  again a high positive association. Hence we may infer high citation counts and higher citations in a paper are positively associated.

\end{enumerate}

Note here that for both network dataset, an assumption of independence among the node-indexed (in-degree, out-degree) counts  perhaps is not perfectly justifiable. Nevertheless, we may consider the data of (in-degree, out-degree) as a random sample from the notional underlying joint degree distribution.
\end{example}

\begin{example}[Danish fire insurance claims data]\label{ex:danish}
Consider the well-studied \emph{Danish Fire insurance} data set obtained from the R package \texttt{fitdistrplus} (\citet{delignette2015fitdistrplus}). The data contains 2167 fire insurance claims at Copenhagen Reinsurance for the period 1980 to 1990 in three different categories: \emph{building}, \emph{content} and \emph{profit} (measured in millions of Danish Krone and inflation adjusted to 1985 values).  
\begin{figure}[ht]
    \centering
    \includegraphics[width=\linewidth]{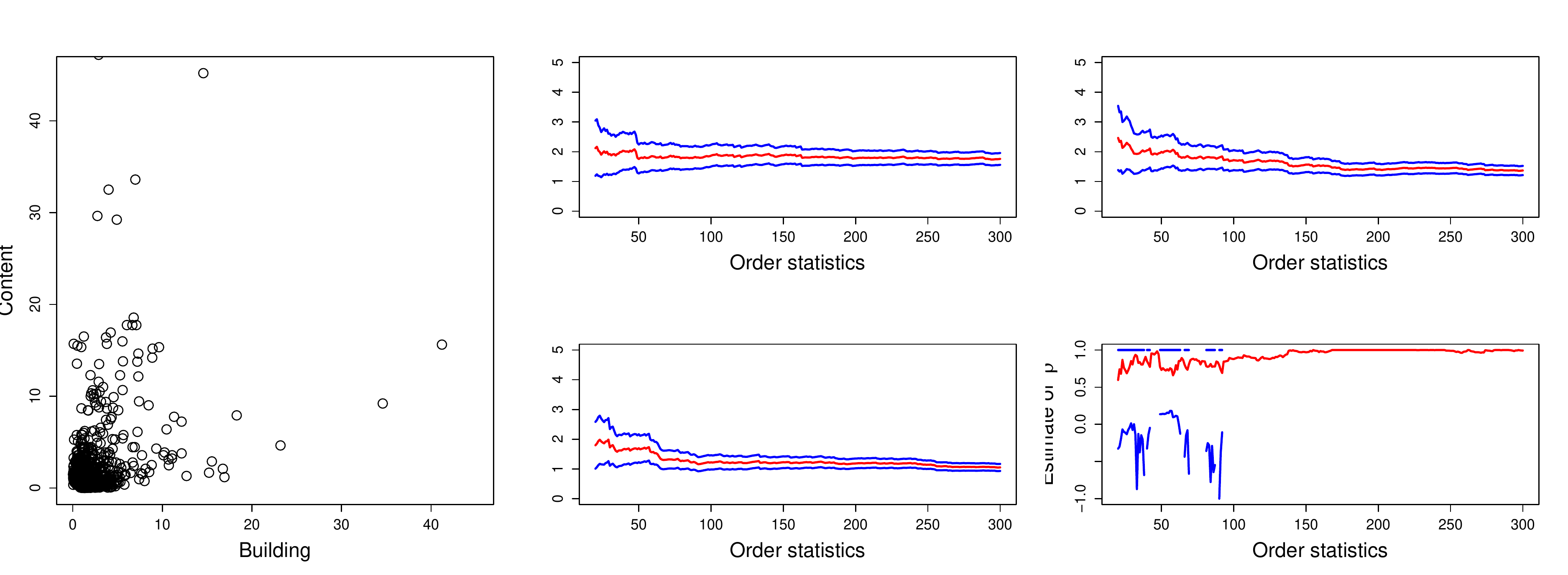} 
 \caption{Danish fire insurance data: (i) Left plot: scatter plot of claim sizes for \underline{building}  vs \underline{content}, (ii) middle plot: Hill estimates of tail indices of \underline{building} (top) and \underline{content} (bottom), (iii) right plot: Hill estimates of tail index of $\min$(building, content) at the top and estimate of $\rho$ (bottom). Plots of estimates are for top 300 quantiles and with 95\% estimated confidence intervals.} 
  \label{fig:danishestimates}
\end{figure}
Focusing on \emph{building} and \emph{content} claims (which were both simultaneously positive in 604 claims), data exploration (heavy-tailed QQ plot) suggests  both marginal distributions to be heavy-tailed, and so do their respective Hill estimates of tail indices $\alpha_1$ and $\alpha_2$ (the values are close to 2 or less). For the data of $\min$(\emph{building, content}), we  also estimate the tail index $\gamma$ to be 2 or slightly above, see the plots in \Cref{fig:danishestimates}. 

The scatter plot (leftmost plot in \Cref{fig:danishestimates}) suggests a reasonable high positive association between claims for \emph{building} and \emph{content}, and we also notice that the estimate for correlation parameter $\rho$ is quite high almost nearing 1 (the confidence intervals behave erratically, which we have observed previously for high correlation values close to the boundary of valid inference for $\rho$ in a simulated data as well, cf. \Cref{fig:frechet} when $\rho=0.8$).
\end{example}

\begin{figure}[ht]
    \centering
    \includegraphics[width=\linewidth]{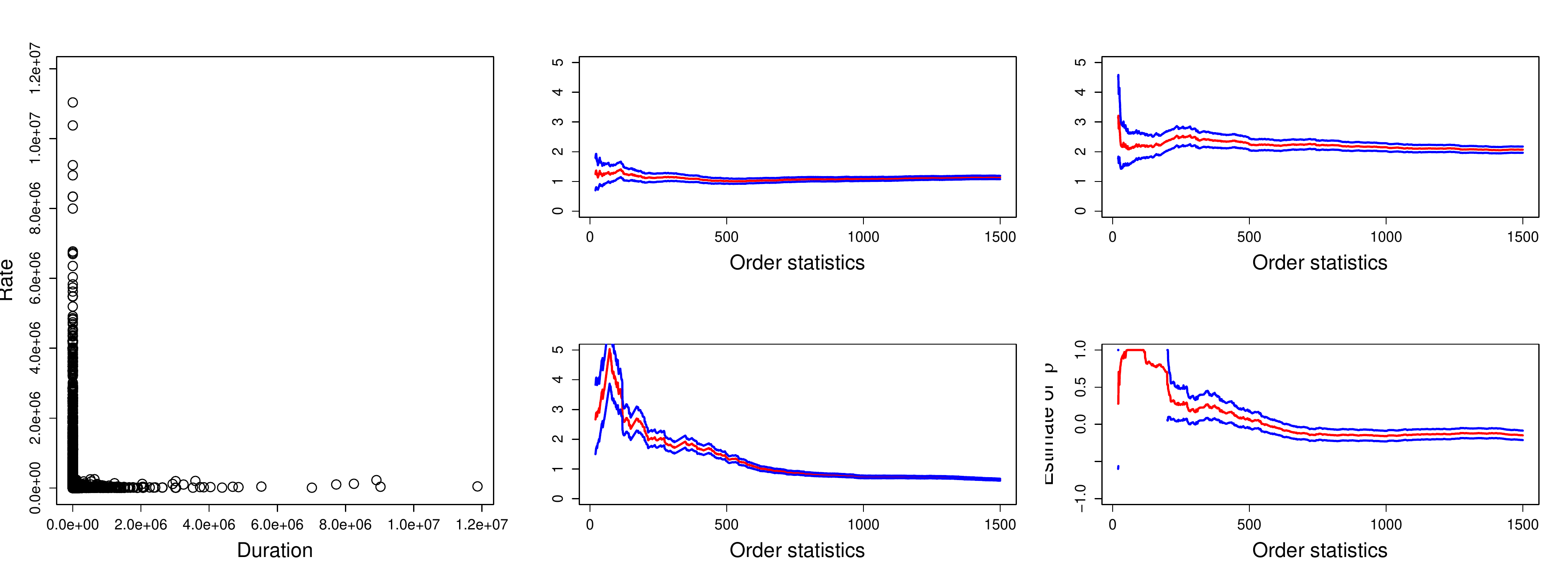} 
 \caption{Internet traffic data: (i) Left plot: scatter plot of \underline{duration} vs \underline{rate}, (ii) middle plot: Hill estimates of tail indices of \underline{duration} (top) and \underline{rate} (bottom), (iii) right plot: Hill estimates of tail index of $\min$(duration, rate) at the top and estimate of $\rho$ (bottom). Plots of estimates are for top 1500 quantiles and with 95\% estimated confidence intervals.} 
  \label{fig:durrate}
\end{figure}
\begin{example}[Internet traffic data]\label{ex:internet}
Heavy-tailed marginals have been observed in internet traffic data for variables such as file size, transmission duration and session length (\cite{maulik:resnick:2003, sarvotham:riedi:baranuik:2005}). We study a particular data set of GPS-synchronized traces that were recorded at
the University of Auckland \href{http://wand.cs.waikato.ac.nz/wits}. We consider traces corresponding exclusively to incoming TCP
traffic sent on December 8, 1999, between 3 and 4 p.m. The packets were clustered
into end-to-end (e2e) sessions which are clusters of packets with the same source and
destination IP address such that the delay between arrival of two successive packets
in a session is at most two seconds. We observe three variables $\{(S_i,D_i,R_i): 1\le i \le 54353\}$ where 
$S_i$ is the size or number of bytes transmitted in a session, $D_i$ is the duration or length of the session, and $R_i=S_i/D_i$ is the  average transfer rate associated with a session.

For this study we consider $D_i$ (duration) and $R_i$ (rate). Diagnostics suggest that both are heavy-tailed and the scatter plot in \Cref{fig:durrate} clearly indicates asymptotic independence between the two variables. Modeling using $\PGC$ we estimate the tail parameters $\alpha_1$ (duration) and $\alpha_2$ (rate); both are reasonable stable close to the value 1. The tail index for $\min$ (duration, rate) is closer to 2 and still quite stable supporting the presence of heavy-tails. For estimating $\rho$, we observe the value to be stable close to 0.

\end{example}

For the examples used in this section, we observe that the plots of Hill estimates do indicate heavy-tailed margins and the stability in estimation of $\rho$ provides support for modeling using a $\PGC$ distribution. The codes and relevant data used for generating the plots are available at \href{https://github.com/bikram-jit-das/ParetoGaussCop}{https://github.com/bikram-jit-das/ParetoGaussCop}.

\section{Conclusion}\label{sec:concl}

In light of the ubiquity of empirical evidence for heavy-tailed variables in various applications and popularity of Gaussian dependence, the proposed $\PGC$ distribution incorporates both for  modeling multivariate heavy-tailed data.
% The traditional modeling and estimation paradigm for such data proceeds by either assuming global marginal tail equivalence or by transforming all margins to the same distribution. This does have a noticeable effect on the dependence, and in turn tail dependence  of the data. Moreover such inverse transformations do affect the quality of convergence of the parameter estimates, although asymptotically they may remain consistent. As an advantage here, the $\PGC$ model proposed allows for non-homogeneous tail modeling and in this paper we establish consistency and asymptotic normality of the  parameter estimates of this model.
Here we provide a few observations on estimation under the $\PGC$ distribution.

\begin{enumerate}[(i)]
    \item We have used Hill estimators for tail index estimation, this is not particularly necessary. There are multiple well-established estimators which may do justice not only in estimating the tail parameters but the correlation parameter as well. The attraction of using Hill estimator stems from the observation that its consistency is equivalent to the regular variation of the tail distribution.
    \item A possible alternative to the $\PGC$  distribution is to use elliptical distributions with a heavy-tailed radial vector (\citet[Chapter 6]{mcneil:frey:embrechts:2015}). A difference from the $\PGC$ model here is that for elliptical distributions, all marginal tails have equivalent tail behavior; a recent work by \citet{derugminy:fermanian:2022} discusses estimation in such elliptical models. Incidentally, with all margins equal, the $\PGC$ distribution is also an elliptical distribution.
    \item If we assume the marginal tails to be the same, then we may use a joint estimator for the tail index, see \citet{dematteo:clemencon:2016}.
    \item Since the correlation parameter is estimated from the relevant margins, the 2-dimensional consistency results suffice for $\PGC$ model parameter estimation in higher dimensions, cf. \Cref{thm:multidimconv}. The key difference of the $d$-dimensional model from the 2-dimensional model is the computation of multiple covariance between various parameter estimates; we leave this for future research.
    \item We observed that a few real life examples support modeling using $\PGC$ distribution; an advantage for this model is also that it is relatively easy to simulate from, and thus helpful for stress-testing and scenario generation. 
    \item Instead of using Pareto-like margins, we may use any of the parametric models from \Cref{table:RVdist} and proceed with modeling and estimation to obtain similar results and insights. Naturally, the other relevant marginal parameters need to investigated for such an endeavor.
    \item Using standard regular variation for approximating tail probabilities like \eqref{prob:tailset} may already lead to degeneracy in the first step if the model has unequal tail indices; hence a subsequent estimation of correlation with such an approach seems infeasible.
\end{enumerate}
Overall, the $\PGC$ distribution provides a general framework for modeling heavy-tailed data and can be used for various risk management applications from modeling to generation and estimation.

\section*{Acknowledgements} BD thanks Vicky Fasen-Hartmann for discussions and detailed comments  on the paper which has significantly improved its exposition.

\bibliographystyle{imsart-nameyear}
\bibliography{bibfilenew}

\appendix

\section{Proofs}\label{sec:proofs}

\subsection{Proof of Theorem \ref{thm:mainrect} }\label{subsec:proofprobext}
We first collate results which are required for the proof of \Cref{thm:mainrect}. The proof is similar to that of \cite[Proposition 3.3]{das:fasen:2023a} with modifications to address the different structure of the marginal distributions.  The first lemma provides the asymptotic behavior of Gaussian quantiles in terms of their equivalent regularly varying quantiles (cf. \citet[Lemma 3.3]{das:fasen:2023a}) and as a result also  quantiles of power law-tailed/ Pareto-tailed distributions.

\begin{lemma}[Lemma 3.3, \citet{das:fasen:2023a}]\label{lem:rvtogaussquantile}
    Let $F_\alpha$ be a strictly increasing and continuous distribution function with distribution function with
    $\ov{F}_{\alpha}(t)=(t^{\alpha}\ell(t))^{-1}$ where $\ell\in \RV_0$ for some $\alpha>0$. Fix $x>0$ and define for any $t>0$  $$z_{x}(t)=\Phi^{-1}(F_\alpha(tx))=\overline{\Phi}^{-1}(\overline{F}_\alpha(tx)).$$ 
    Then as $t\to\infty$,
    \begin{align*}%\label{eq:quantilesparnorm}
     z_x(t) & =  \sqrt{2\alpha \log t}+\frac{1}{\sqrt{2\alpha\log t}}\log\left(\frac{x^{\alpha}}{2\sqrt{\pi\alpha}}\right)  +  \frac{ \log \left(\ell(t)/\sqrt{\log t}\right)}{\sqrt{2\alpha\log t}} + o\left(\frac1{\sqrt{\log t}}\right).
    \end{align*} 
   % \begin{align}\label{eq:quantilesparnorm}
   %   z_x(t)  =  \sqrt{2\alpha \log (t)} + &  \frac{\sqrt{2\alpha}}{\sqrt{\log (t)}}\left[\frac12 \log(\ell(t))-\frac1{4\alpha}\log(\log(t))\right] \nonumber\\ + & \frac{\sqrt{2\alpha}}{\sqrt{\log (t)}}\left[\frac12 \log(x)-\frac1{4\alpha}\log(4\pi\alpha)\right] + o\left(\frac1{\sqrt{\log(t)}}\right).
  %  \end{align} 
%    where w.l.o.g. we can write $\overline{F}_\alpha(t) = , t>0$, for some slowly varying function $\ell$.
  \end{lemma}

\begin{corollary}[Corollary to \Cref{lem:rvtogaussquantile}]\label{cor:paretotogaussquantile}
     Let $F_\alpha$ be a strictly increasing and continuous distribution function with
    $\ov{F}_{\alpha}(t)\sim C t^{-\alpha}$ as $t\to\infty$ where $C>0, \alpha>0$. Fix $x>0$ and define for any $t>0$  $$z_{x}(t)=\Phi^{-1}(F_\alpha(tx))=\overline{\Phi}^{-1}(\overline{F}_\alpha(tx)).$$ 
    Then as $t\to\infty$,
    \begin{align*}%\label{eq:quantilesparnorm}
       z_x(t)  =  \sqrt{2\alpha \log t}+\frac{1}{\sqrt{2\alpha\log t}}\log\left(\frac{x^{\alpha}}{2C\sqrt{\pi\alpha}}\right)  -  \frac{ \log\log t}{2\sqrt{2\alpha\log t}} + o\left(\frac1{\sqrt{\log t}}\right).
       \end{align*}
\end{corollary}
\begin{proof}
The proof follows from \Cref{lem:rvtogaussquantile} by noticing $\ell(t)\sim C$ ($t\to\infty$) in the tail behavior of $F$.
\end{proof}

The following lemma provides the solution to a quadratic program essential for defining and proving \Cref{thm:mainrect}, see \cite{hashorva:husler:2003,hashorva:2005,hashorva:2019} for details and proof.
  \begin{lemma}[Lemma 1, \citet{hashorva:2019}]\label{lem:qp}
Let $\Sigma\in \R^{d\times d}$ be a positive definite correlation matrix and let $\bc\in(-\binfty,\bzero]^c$. Then the quadratic programming problem 
\begin{align}\label{eq:quadprog}
  \mathcal{P}_{\Sigma^{-1},\bc}: \min_{\{\bz\ge \bc\}} \bz^\top \Sigma^{-1}\bz  
\end{align}
  has a unique solution $\bkappa\in\R^d$ such that
    \begin{align}
    \gamma_{\bc}:= \gamma(\Sigma,\bc):= \min_{\{\bz\ge \bc\}} \bz^\top \Sigma^{-1} \bz=\bkappa^{\top} \Sigma^{-1}\bkappa>0. 
\end{align}
Moreover, the following holds:
\begin{enumerate}[(i)]
    \item There exists a unique non-empty index set  $I:=I(\Sigma,\bc)\subseteq \{1,\ldots, d\}=:\mathbb{I}$ with $J:=J(\Sigma,\bc):=\mathbb{I}\setminus I$ such that  the unique solution $\bkappa$ is given by
$$\bkappa_I=\bc_I, \qquad \text{ and } \qquad\bkappa_J=-[\Sigma^{-1}]_{JJ}^{-1}[\Sigma^{-1}]_{JI}\bc_I=\Sigma_{JI}(\Sigma_I)^{-1}\bc_I\ge \bc_J,$$
and, $ \bc_{I}\Sigma_{I}^{-1}\bc_I=\bkappa^{\top} \Sigma^{-1}\bkappa=\gamma_{\bc}.$
 \item Define $\bh_I:=(h_i)_{i\in I}$ where $h_i=h_i(\Sigma,\bc):=e_i^\top \Sigma^{-1}_{I}\bc_I>0, i\in I$. Then for any $\bz\in\R^d$, $$\bz^{\top}\Sigma^{-1}\bkappa=\bz_I^{\top}\Sigma_{I}^{-1}\bc_I = \bz^{\top}\bh_{I}.$$
\end{enumerate}
 %the following equality holds: $$\bz^{\top}\Sigma^{-1}\bkappa=\bz_F^{\top}\Sigma_{F}^{-1}\bc_F,$$
% for any  $F\supset I$.
\end{lemma}

Given the solution and notations for the quadratic program $\mathcal{P}_{\Sigma,\bc}$ we are able to compute tail probabilities of the form \eqref{prob:tailset} where $S=\mathbb{I}$; for any other non-empty $S\subsetneq \mathbb{I}$, we just consider the $|S|$-dimensional  marginal vector as our notional random vector to be analyzed.

The following result gives asymptotic probability of tail exceedance for a multivariate normal distribution; the tail exceedance has a particular structure amenable for the proof of  \Cref{thm:mainrect} as well as \Cref{prop:condR}. Here we use a particular convention: for sequences $\bx_t=(x_{t,1},\ldots,x_{t,d}),\by_t=(y_{t,1},\ldots,y_{t,d})\in \R^d$, we write  $\by_t=\bx_t+o(t)$ if $y_{t,j}-x_{t,j}=o(t)$ as   $t\to\infty$  for all  $j\in \mathbb{I}$.

\begin{proposition}\label{prop:gausstail}
    Let $\bZ\sim \Phi_{\Sigma}$ be a normal random vector in $\R^d$, $d\ge 2$, with positive definite correlation matrix $\Sigma\in\R^{d\times d}$ and $\bc >\bzero$.  Define the following quantities:
    \begin{enumerate}[(a)]
    \item The parameters $\gamma=\gamma(\Sigma,\bc), I=I(\Sigma, \bc)$, $\bkappa=\bkappa(\Sigma,\bc)$ and $h_i=h_i(\Sigma, \bc)$, $i\in I$ are defined with respect to the solution of  $\mathcal{P}_{\Sigma^{-1}, \bc}$ as in \Cref{lem:qp}.   
    \item Let  $J:=\mathbb{I}\setminus I$. Define
 $\bY_J\sim \mathcal{N}(\bzero_J, \Sigma_J-\Sigma_{JI}\Sigma_{I}^{-1}\Sigma_{IJ})$ if $J\not=\emptyset$ and $\bY_J=\bzero$ if $J=\emptyset$. 
 \item Define $\bl_{\infty}:=\lim_{t\to\infty}t(\bc_J-\bkappa_J)$, a vector in $\R^{|J|}$ with components either  $0$ or $-\infty$.
  \end{enumerate}
  Let $\bz, \bw \in \R^d$ and  $\mathcal{L},u:(0,\infty)\to(0,\infty)$ be measurable functions such that $u(t)\uparrow \infty$ and ${\log(\mathcal{L}(t))}/{u(t)}\to 0,$  as $t\to\infty$.
    Then as $t\to \infty$,
    \begin{align*}
         \P & \left(\bZ >  u(t)\bc + \frac{\bz}{u(t)}  + \frac{\log(\mathcal{L}(t))}{u(t)}\bw +o\left(\frac{1}{u(t)}\right)\right) \nonumber\\
          & \quad\quad\quad\quad\quad =(1+o(1))\Upsilon u(t)^{-|I|} (\mathcal{L}(t))^{-\bw_I^T\Sigma_I^{-1}\bc_I} \exp\left(-\gamma \frac{u(t)^2}2 - \bz_I^\top\Sigma_I^{-1}\bc_I\right),%\label{eq:multtail}
    \end{align*}
where 
\begin{align*}%\label{eq:kappaSigma}
  \Upsilon:=\Upsilon(\Sigma,\bc):=
  {(2\pi)^{-|I|/2}|\Sigma_{I}|^{-1/2}\prod_{i\in I}h_i^{-1} }\,\times {\P(\bY_J\geq \bl_\infty)}.
\end{align*}
For $J=\emptyset$ we define $\P(\bY_J\geq \bl_\infty):=1$.
\end{proposition}

\begin{proof}
Using \citet[Corollary~3.3]{hashorva:2005} %with  $\lim_{t\to\infty}t(\bone-\be(\Sigma^{-1})^*)_J=(-\infty,\ldots,\infty)$  
we know that for  $\bx\in \R^d$, as $t\to\infty$,
\begin{align*}
        \P(\bZ & >u(t)\bc + \bx) \\
        & =(1+o(1)) {(2\pi)^{-|I|/2}|\Sigma_{I}|^{-1/2}\prod_{i\in I}h_i^{-1} }\times{\P(\bY_J>\bl_\infty+\bx_J-\Sigma_{JI}\Sigma_I^{-1}\bx_I)} \\
        &\qquad\qquad\qquad \times u(t)^{-|I|}\exp\left(-\gamma \frac{u(t)^2}2 - u(t)\bx^\top\Sigma^{-1}\bkappa-\frac12\bx^{\top}\Sigma^{-1}\bx\right).      
    \end{align*}
    Define $\bx^{(t)}:=\frac{\bz}{u(t)}  + \frac{\log(\mathcal{L}(t))}{u(t)}\bw +o\left(\frac{1}{u(t)}\right)$. Hence, $\lim_{t\to\infty}\bx^{(t)}=\bzero$ and we have $\bkappa_J\geq\bc_J$ by \Cref{lem:qp}.
    Thus,
    \beao
    \lim_{t\to\infty}\P\left(\bY_J\ge  u(t)(\bc_J-\bkappa_J)+\bx_J^{(t)}-\Sigma_{JI}\Sigma_I^{-1}\bx_I^{(t)}\right)=\P(\bY_J\ge \bl_\infty).
\eeao
Following arguments  analogous to the proof of Theorem~3.1 and Corollary~3.3 
in \citet{hashorva:2005}, we are allowed to replace  $\bx$ by $\bx^{(t)} $ and
$\P(\bY_J>\bl_\infty+\bx_J-\Sigma_{JI}\Sigma_I^{-1}\bx_I)$ by $\P(\bY_J\ge \bl_\infty)$
to   obtain
\begin{align*}
         \P\Bigg(\bZ > u(t)\bc & + \frac{\bz}{u(t)}  + \frac{\log(\mathcal{L}(t))}{u(t)}\bw +o\left(\frac{1}{u(t)}\right) \Bigg)\\
         &= \P\left(\bZ > u(t)\bc + \bx^{(t)}\right)\\ 
        &=(1+o(1))\Upsilon u(t)^{-|I|}  
        \exp\left(-\gamma \frac{u(t)^2}2 - u(t)\bx^{(t)\top}\Sigma^{-1}\bkappa-\frac12\bx^{(t)\,\top}\Sigma^{-1}\bx^{(t)}\right) 
\intertext{and since $\bx^{(t)}\to \bzero$ and $u(t)\bx^{(t)}\sim \bz  +\log(\mathcal{L}(t))\bw$ as $t\to\infty$ we have the above to be} 
        &=(1+o(1))\Upsilon u(t)^{-|I|}  
        \exp\left(-\gamma \frac{u(t)^2}2 - \bz^\top\Sigma^{-1}\bkappa- \log{\mathcal{L}(t)}\cdot\bw^{\top}\Sigma^{-1}\bkappa \right) \\  
       & =(1+o(1))\Upsilon u(t)^{-|I|} (\mathcal{L}(t))^{-\bw_I\Sigma_{I}^{-1}\bc_I} \exp\left(-\gamma \frac{u(t)^2}2 - \bz_I^\top\Sigma_I^{-1}\bc_I\right),
    \end{align*}
  where in the last step using \Cref{lem:qp} we have $\bw^{\top}\Sigma^{-1}\bkappa= \bw_I\Sigma_{I}^{-1}\bc_I$ and  $\bz^{\top}\Sigma^{-1}\bkappa= \bz_I\Sigma_{I}^{-1}\bc_I$ .  
 \end{proof}

 \begin{proof}[Proof of \Cref{thm:mainrect}] 
Let $\bX\sim F$ so that $X^i\sim F_i$ with $\ov{F}_i(t) \sim \theta_it^{-\alpha_i}, t\to\infty$, for $i\in \mathbb{I}$.  First using \Cref{cor:paretotogaussquantile}, we have
\begin{align*}
     \P(X_i > tx_i, \forall i \in \mathbb{I}) = \P\left(Z_i > z_{x_i}(t), \forall i \in \mathbb{I} \right),
\end{align*}
where $z_{x_i}(t)=\ov \Phi^{-1}(\ov F_{i}(t x_i))$, $i\in \mathbb{I}$, and $\bZ \sim \Phi_{\Sigma}$. %By tail equivalence we have $\ov{F}_s(t)\sim \ov{F}_1(t) = (t^{\alpha}\ell(t))^{-1}$. %Moreover, there exists a slowly varying function $\ell\in \RV_0$ with $\overline{F}_j(t) \sim \overline{F}_1(t)=\frac{1}{t^{\alpha}\ell(t)}$,  as $t\to\infty$, $\forall j\in \mathbb{I}$, since the $\ov F_j(t)/\ov F_1(t)\to 1$ as $t\to\infty$ for $j=1,\ldots,d$ by assumption. 
Defining $\bz := \left(\frac{1}{\sqrt{\alpha_i}}\log\left(\frac{x_i^{\alpha_i}}{2\theta_i\sqrt{\pi\alpha_i}}\right)\right)_{i\in \mathbb{I}}$ and applying \Cref{cor:paretotogaussquantile} we get 
\begin{align} 
    &  \P(X_i > tx_i, \forall i \in \mathbb{I}) \nonumber \\
    & = \P\left(\bZ> \sqrt{2\log t} \sqrt{\balpha} + \frac{\bz}{\sqrt{2\log t}} - \frac{\log\log t}{\sqrt{\log t}}\frac{\sqrt{\balpha^{-1}}}{2}  + o\left(\frac{1}{\sqrt{\log t}}\right)\right) \nonumber\\
                  & = \P\left(\bZ> u(t)  \sqrt{\balpha} + \frac{\bz}{u(t)} - \frac{\log\mathcal{L}(t)}{u(t)}\bw + o\left(\frac{1}{u(t)}\right)\right), \label{eq4}
\end{align}
where $u(t)=\sqrt{2\log t}, \mathcal{L}(t)=\log t$ and $\bw = -\frac{\sqrt{\balpha^{-1}}}{2}$. Since we have $$\lim_{t\to\infty}\log(\mathcal{L}(t))/u(t)= \lim_{t\to\infty}\log(\log t)/\sqrt{\log t} = 0, $$  using \Cref{prop:gausstail}
and denoting $h_i= e_i^{\top}\Sigma_{I}^{-1}\sqrt{\alpha_I}, \Delta= (\sqrt{\balpha_{I}}^{-1})^{\top}\Sigma_{I}^{-1}\sqrt{\balpha_{I}}$,  as $t\to \infty$, we have
\begin{align*}
    \P(X_i > tx_i, \forall i \in \mathbb{I})  & = (1+o(1)) \Upsilon(\sqrt{2\log t})^{-|I|} (\log t)^{\frac{\Delta}{2}}\exp\left(-\gamma \log t- \bz_I^{\top}\Sigma_I^{-1}\sqrt{\balpha_I}\right)\\
     & = (1+o(1))\Upsilon 2^{-\frac{|I|}{2}} (\log t)^{\frac{\Delta-|I|}{2}} t^{-\gamma} \prod_{i\in I} \exp\left(-\frac{h_i}{\sqrt{\alpha_i}}\log\left(\frac{x_i^{\alpha_i}}{2\theta_i\sqrt{\pi\alpha_i}}\right)\right)\\
             & = (1+o(1))\Upsilon 2^{-\frac{|I|}{2}} (2\sqrt{\pi})^{\Delta} t^{-\gamma} (\log t)^{\frac{\Delta-|I|}{2}} \prod_{i\in I} (\theta_i\sqrt{\alpha_i})^{\frac{h_i}{\sqrt{\alpha_i}}} \prod_{i\in I} x_i^{-\sqrt{\alpha_i}h_i}\\
             & = (1+o(1))\Psi t^{-\gamma} (\log t)^{\frac{\Delta-|I|}{2}}\prod_{i\in I} x_i^{-\sqrt{\alpha_i}h_i}
\end{align*}
where $\Psi$ is as defined in \eqref{def:Psis}.
\end{proof}

\subsection{Proof of Proposition \ref{prop:condR}}\label{subsec:proofcondR}

\begin{proof}
For $i=1,2$, we have $\ov{F}_i(t)\sim \theta_it^{-\alpha_i}$ as $t\to\infty$ which we can check is equivalent to  $\ov{F}_i^{\leftarrow}(w/t)\sim (C/w)^{1/\alpha_i} t^{1/\alpha_i}$ as $t\to\infty$ for $w>0$. Let $\bZ= (Z^1, Z^2) \sim \Phi_{\Sigma}$.
    \begin{enumerate}[(i)]
        \item For computing $R_{1,2}(w_1,w_2)$ observe that,
        \begin{align*}
           & t\, \P\left(\ov{F}_1(X^1)\le \frac {w_1}t, \ov{F}_2(X^2)\le \frac {w_2}t\right) \\
           & =  t\, \P\left(X^1 \ge \ov{F}^{\leftarrow}_1\left(\frac {w_1}t\right),X^2 \ge \ov{F}^{\leftarrow}_2\left(\frac {w_2}t\right)\right)\\
           & \sim t\, \P\left(X^1 \ge (\theta_1/w_1)^{1/\alpha_1} t^{1/\alpha_1} ,X^2 \ge (\theta_2/w_2)^{1/\alpha_2} t^{1/\alpha_2} \right)\\
      %    & = \P\left(Z^i> \sqrt{2\log t} - \frac{1}{\sqrt{2\log t}} \log\left(\frac{1}{2x\sqrt{\pi\alpha_i}}\right) - \frac{\log\log t-\log{\alpha_1}\}{2\sqrt{2\log t}}  + o\left(\frac{1}{\sqrt{\log t}}\right), i=1,2\right) \nonumber\\
      & = t\, \P\left(Z^i> \sqrt{2\log t} - \frac{1}{\sqrt{2\log t}} \log\left(\frac{1}{2x\sqrt{\pi}}\right) - \frac{\log\log t}{2\sqrt{2\log t}}  + o\left(\frac{1}{\sqrt{\log t}}\right), i=1,2\right) \nonumber\\  
     %   & = \P\left(\bZ >  u(t)\bc + \frac{\bz}{u(t)}  + \frac{\log(\mathcal{L}(t))}{u(t)}\bw +o\left(\frac{1}{u(t)}\right)\right) 
     & = t \times (1+o(1))t^{-\frac{2}{1+\rho}} (\log t)^{-\frac{1}{1+\rho}}  \eta_{w_1,w_2} 
        \end{align*}
   for some $\eta_{w_1,w_2}>0$ using \Cref{prop:gausstail}.  The approximation ``$\sim$" above in the second step can be formalized by bounding the quantity with $\theta_j\pm \epsilon$.  Now since $\rho\in (-1,1)$ we immediately obtain have $R(w_1,w_2)$=0. %Note that here we did not need $\rho<\min(\sqrt{\alpha_1/\alpha_2},\sqrt{\alpha_1/\alpha_2})$.
        \item %Since by definition we can check that $\gamma > \max(\alpha_1,\alpha_2)$,  hence 
        %\[ \rho < \min(\sqrt{\alpha_1/\gamma},\sqrt{\alpha_2/\gamma}) < \min(\sqrt{\alpha_1/\alpha_2},\sqrt{\alpha_1/\alpha_1})\]
        From the proof of \Cref{thm:2dimconv}, for $\rho<\min(\sqrt{\alpha_1/\alpha_2},\sqrt{\alpha_1/\alpha_2})$,  we have $$\ov{G}(t)\sim \Psi t^{-\gamma} (\log t)^{-\delta}, \quad  t\to\infty$$ where $\gamma=\frac{\alpha_1+\alpha_2-2\rho\sqrt{\alpha_1\alpha_2}}{(1-\rho^2)}$ and $\delta = \frac{\rho\gamma}{\sqrt{\alpha_1\alpha_2}}$. This we can check is equivalent to
        \begin{align}
            \ov{G}^{\leftarrow}(w/t)\sim (\Psi\gamma^{\delta}/w)^{1/\gamma} t^{1/\gamma} (\log t)^{\delta}, \quad t\to\infty. \label{eq:ovGinv}
        \end{align}
       Hence for computing $R_{1,12}(w_1,w_2)$ observe that,
        \begin{align}
           & t\, \P\left(\ov{F}_1(X^1)\le \frac {w_1}t, \ov{G}(Y)\le \frac {w_2}t\right) \nonumber \\
           & =  t\, \P\left(X^1 \ge \ov{F}^{\leftarrow}_1\left(\frac {w_1}t\right),\min(X^1, X^2) \ge \ov{G}^{\leftarrow}\left(\frac {w_2}t\right)\right) \nonumber \\
           & \sim t\, \P\Big(X^1 \ge \max\left(t^{1/\alpha_1} (\theta_1/w_1)^{1/\alpha_1}, (\Psi\gamma^{\delta}/w_2)^{1/\gamma} t^{1/\gamma} (\log t)^{\delta}\right), \nonumber \\
           &\qquad\qquad\qquad\qquad\qquad\qquad\qquad\qquad\qquad X^2 \ge (\Psi\gamma^{\delta}/w_2)^{1/\gamma} t^{1/\gamma} (\log t)^{\delta}\Big) \nonumber 
 \intertext{using \eqref{eq:ovGinv} and since $\gamma>\alpha_1$, we have $(\theta_1/x)^{1/\alpha_1}>(\Psi\gamma^{\delta}/w)^{1/\gamma} t^{1/\gamma} (\log t)^{\delta}$ eventually and hence the above}
      %    & = \P\left(Z^i> \sqrt{2\log t} - \frac{1}{\sqrt{2\log t}} \log\left(\frac{1}{2x\sqrt{\pi\alpha_i}}\right) - \frac{\log\log t-\log{\alpha_1}\}{2\sqrt{2\log t}}  + o\left(\frac{1}{\sqrt{\log t}}\right), i=1,2\right) \nonumber\\
      & = t\, \P\left(X^1 \ge (\theta_1/w_1)^{1/\alpha_1}t^{1/\alpha_1} ,X^2 \ge (\Psi\gamma^{\delta}/w_2)^{1/\gamma} t^{1/\gamma} (\log t)^{\delta}\right). \label{eq:RX1Y1mid}
      \end{align}
   \emph{Case 1:} Let $\rho<\sqrt{\frac{\alpha_2}{\gamma}}$. Now from  \eqref{eq:RX1Y1mid} we have
      \begin{align*}
        & t\, \P\left(\ov{F}_1(X^1)\le \frac {w_1}t, \ov{G}(Y)\le \frac {w_2}t\right) \\
      & < t\, \P\left(X^1 \ge (\theta_1/w_1)^{1/\alpha_1}t^{1/\alpha_1} ,X^2 \ge (\Psi\gamma^{\delta}/w_2)^{1/\gamma} t^{1/\gamma}\right)\\
       & = t\, \P\Bigg(Z_1> \sqrt{2\log t} - \frac{1}{\sqrt{2\log t}} \log\left(\frac{1}{2w_1\sqrt{\pi}}\right) - \frac{\log\log t}{2\sqrt{2\log t}}  + o\left(\frac{1}{\sqrt{\log t}}\right),  \nonumber\\
        & \qquad    Z_2> \sqrt{\frac{\alpha_2}{\gamma}}\sqrt{2\log t} - \frac{\sqrt{\gamma/\alpha_2}}{\sqrt{2\log t}} \log\left(\frac{\left(\psi\gamma^{\delta}/w_2\right)^{\alpha_2/\gamma}}{2\theta_2\sqrt{\pi\alpha_2/\gamma}}\right) - \frac{\log\log t}{2\sqrt{2\log t}}  + o\left(\frac{1}{\sqrt{\log t}}\right) \Bigg)\\
         %   & = \P\left(\bZ >  u(t)\bc + \frac{\bz}{u(t)}  + \frac{\log(\mathcal{L}(t))}{u(t)}\bw +o\left(\frac{1}{u(t)}\right)\right) 
     & = t \times (1+o(1))t^{-\frac{\alpha_2+ \gamma-2\rho\sqrt{\alpha_2\gamma}}{\gamma(1-\rho^2)}} (\log t)^{-\frac{\rho(\alpha_2+ \gamma)-2\rho^2\sqrt{\alpha_2\gamma}}{2\sqrt{\alpha_2\gamma}(1-\rho^2)}}  \eta^*_{w_1,w_2} 
        \end{align*}
   for some $\eta^*_{w_1,w_2}>0$ using \Cref{prop:gausstail} since $\rho<\sqrt{\frac{\alpha_2}{\gamma}}$. We can check that  $\frac{\alpha_2+ \gamma-2\rho\sqrt{\alpha_2\gamma}}{\gamma(1-\rho^2)}>1$, and hence by taking limits above we get $R_{1,12}=0$.\\
\noindent \emph{Case 2:} Let $\rho=\sqrt{\frac{\alpha_2}{\gamma}}$. Using \Cref{prop:gausstail} and solving the quadratic program $\mathcal{P}_{\Sigma^{-1}, (1,\sqrt{\alpha_2/\gamma})}$ we can check by following similar arguments as in the proof of \Cref{cor:2dim} that
 \[R_{1,12} = w_1/2.\]
 \noindent \emph{Case 3:} Let $\rho>\sqrt{\frac{\alpha_2}{\gamma}}$. The result follows again in a similar manner as above.
   \item The limits of $R_{2,12}$ can be found by following similar arguments as in proving \eqref{eq:RX1Y}.
   
    \end{enumerate}
\end{proof}

\begin{remark}
    The conclusion in \eqref{eq:RX1X2} can also be shown  using asymptotic independence of the Gaussian copula, cf. \cite{hua:joe:2011, das:fasen:2018}. The other two conclusions require subtler computation as exhibited above.
\end{remark}

\section{Auxiliary result for Section \ref{sec:parmeterest}}\label{sec:aux}

\begin{lemma}\label{lem:solverho}
    Let $\hat{\alpha}_1>0, \hat{\alpha}_2>0, \hg>0$. If $\rho\in (-1,1)$ and  $\rho<\min\left(\sqrt{{\hat{\alpha}_2}/{\hat{\alpha}_1}}, \sqrt{{\hat{\alpha}_1}/{\hat{\alpha}_2}}\right)$, then the unique solution to 
    \begin{align}\label{eq:quad}
      g(\rho)= \hg\rho^2-2\sqrt{\hat{\alpha}_1\hat{\alpha}_2}\rho+(\hat{\alpha}_1+\hat{\alpha}_2-\hg)=0
     \end{align}
    is given by
    \[\hat{\rho} = \left({\sqrt{\hat{\alpha}_1\hat{\alpha}_2}-\sqrt{\hat{\alpha}_1\hat{\alpha}_2+\hg^2-\hg(\hat{\alpha}_1+\hat{\alpha}_2)}}\right)\Big/{\hg}.\]
\end{lemma}

\begin{proof}
The quadratic equation \eqref{eq:quad} has two solutions:
\begin{align*}
\hat{\rho} & = \left({\sqrt{\hat{\alpha}_1\hat{\alpha}_2}-\sqrt{\hat{\alpha}_1\hat{\alpha}_2+\hg^2-\hg(\hat{\alpha}_1+\hat{\alpha}_2)}}\right)\Big/{\hg},\\
\tilde{\rho} &= \left({\sqrt{\hat{\alpha}_1\hat{\alpha}_2}+\sqrt{\hat{\alpha}_1\hat{\alpha}_2+\hg^2-\hg(\hat{\alpha}_1+\hat{\alpha}_2)}}\right)\Big/{\hg}.
\end{align*}
W.l.o.g. assume $\hat{\alpha}_1\ge\hat{\alpha}_2>0$. Both solutions $\hat{\rho}$ and $\tilde{\rho}$ are real-valued as %the discriminant of $g$ is 
\begin{align*}
  D(g) & = \hat{\alpha}_1\hat{\alpha}_2+\hg^2-\hg(\hat{\alpha}_1+\hat{\alpha}_2) = \hat{\alpha}_1\hat{\alpha}_2+\hg(\hg-\hat{\alpha}_1-\hat{\alpha}_2)\\
      & = \hat{\alpha}_1\hat{\alpha}_2+ \left(\hat{\alpha}_1+ {\left(\sqrt{\hat{\alpha}_2}-\rho\sqrt{\hat{\alpha}_1}\right)^2}\Big/({1-\rho^2})\right)\left({\left(\sqrt{\hat{\alpha}_2}-\rho\sqrt{\hat{\alpha}_1}\right)^2}\Big/({1-\rho^2}) -\hat{\alpha}_2\right)\\
      & = \frac{\left(\sqrt{\hat{\alpha}_2}-\rho\sqrt{\hat{\alpha}_1}\right)^4}{(1-\rho^2)^2} + (\alpha_1-\alpha_2)\frac{\left(\sqrt{\hat{\alpha}_2}-\rho\sqrt{\hat{\alpha}_1}\right)^2}{(1-\rho^2)} \ge 0 \quad (\text{since  w.l.o.g. $\hat{\alpha}_1\ge\hat{\alpha}_2$})
\end{align*}
 Now, we need to show that  $\tilde{\rho}$ cannot be a valid solution to $g(\rho)=0$.  By our conditions, $\rho$ must satisfy $-1<\rho<\sqrt{{\hat{\alpha}_2}/{\hat{\alpha}_1}}$. Since $\tilde{\rho}>0$ always, it precludes solutions for which $-1<\rho<0$. Now if $0\le \rho<\sqrt{{\hat{\alpha}_2}/{\hat{\alpha}_1}}$, then rewriting $g(\rho)=0$ we have
\begin{align}
    \hg  = \frac{\hat{\alpha}_1+\hat{\alpha}_2-2\rho\sqrt{\hat{\alpha}_1\hat{\alpha}_2}}{1-\rho^2}  & = \hat{\alpha}_1\left[1+{\left(\sqrt{\frac{\hat{\alpha}_2}{\hat{\alpha}_1}}-\rho\right)^2}\Bigg/({1-\rho^2})\right] \nonumber\\ 
        & < \hat{\alpha}_1\left[1+{\left(\sqrt{\frac{\hat{\alpha}_2}{\hat{\alpha}_1}}-\rho\right)^2}\Bigg/\left({\frac{\hat{\alpha}_2}{\hat{\alpha}_1}-\rho^2}\right)\right]  = \frac{2\sqrt{\hat{\alpha}_1\hat{\alpha}_2}}{\sqrt{\frac{\hat{\alpha}_2}{\hat{\alpha}_1}}+\rho}. \label{ineq:forgamma}
\end{align}
The inequality above holds since $\hat{\alpha}_2<\hat{\alpha}_1$  and $\rho<\sqrt{{\hat{\alpha}_2}/{\hat{\alpha}_1}}$. Hence, if $\tilde{\rho}$ is a valid solution to \eqref{eq:quad}, then using \eqref{ineq:forgamma} we have
\begin{align*}
    \rho & \ge \frac{\sqrt{\hat{\alpha}_1\hat{\alpha}_2}}{\hg}> \frac12\left(\sqrt{\frac{\hat{\alpha}_2}{\hat{\alpha}_1}}+\rho\right)   
\end{align*}
which implies $\rho>\sqrt{{\hat{\alpha}_2}/{\hat{\alpha}_1}}$, a contradiction to our assumption. Hence $\hat{\rho}$ is the unique valid solution to \eqref{eq:quad}.
%If $\rho\le 0$, then $\tilde{\rho}$ cannot be a correct solution since $\tilde{\rho}>0$.   If  $0\le \rho < \sqrt{\hat{\alpha}_2/\hat{\alpha}_1}$ then .... need to check tomorrow.
\end{proof}

\end{document}